\documentclass{article}
\usepackage[utf8]{inputenc}
\usepackage{amssymb}
\usepackage{hyperref}
\usepackage{mathtools}
\usepackage{amsmath, amsthm}
  {
      
      \newtheorem{proposition}{Proposition}
      
      \newtheorem{lemma}{Lemma}
  }
\usepackage{color}
\usepackage{multirow}
\usepackage{subcaption}
\usepackage[sorting=none]{biblatex}
\def\s[#1\s]{\begin{align}\begin{split}#1\end{split}\end{align}}
\DeclareMathOperator{\e}{e}

\allowdisplaybreaks
\begin{document}

\begin{titlepage}
\title{
\hfill\parbox{4cm}{ \normalsize YITP-21-77}\\ 
\vspace{1cm} 
Counting tensor rank decompositions}
\author{Dennis Obster\footnote{dennis.obster@yukawa.kyoto-u.ac.jp} \ 
and Naoki Sasakura\footnote{sasakura@yukawa.kyoto-u.ac.jp} \\
{\small{\it Yukawa Institute for Theoretical Physics, Kyoto University,}}
\\ {\small{\it  Kitashirakawa, Sakyo-ku, Kyoto 606-8502, Japan}}
}

\date{February 2021}

\maketitle

\begin{abstract}
The tensor rank decomposition is a useful tool for the geometric interpretation of the tensors in the canonical tensor model (CTM) 
of quantum gravity. In order to understand the stability of this interpretation, 
it is important to be able to estimate how many tensor rank decompositions can approximate a given tensor.
More precisely, finding an approximate symmetric tensor rank decomposition of a symmetric tensor $Q$ with an error allowance $\Delta$ is to find vectors $\phi^i$ satisfying $\|Q-\sum_{i=1}^R \phi^i\otimes \phi^i\cdots \otimes \phi^i\|^2 \leq \Delta$. 
    The volume of all possible such $\phi^i$ is an interesting quantity which measures the amount of possible decompositions for a tensor $Q$ within an allowance. 
    While it would be difficult to evaluate this quantity for each $Q$, we find an explicit formula for a similar quantity by integrating over all $Q$ of unit norm.
    The expression as a function of $\Delta$ is given by the product of a hypergeometric function and a power function. 
    We also extend the formula to generic decompositions of non-symmetric tensors. 
    The derivation depends on the existence (convergence) of the partition function of a matrix model which appeared in the context of the CTM.
    
\end{abstract}

\end{titlepage}

\section{Introduction}
The canonical tensor model (CTM) is a tensor model for quantum gravity which is constructed in the canonical formalism in order to introduce time into a tensor model~\cite{Sasakura:2011sq}, with as its fundamental variables the canonically conjugate pair of real symmetric tensors of degree three, $Q_{abc}$ and $P_{abc}$. Interestingly, under certain algebraic assumptions this model has been found to be unique~\cite{Sasakura:2012fb}. Furthermore, several remarkable connections have been found between the CTM and general relativity~\cite{Sasakura:2014gia,Sasakura:2015pxa, Chen:2016ate}, which, combined with the fact that defining the quantised model is mathematically very simple and straightforward~\cite{Sasakura:2013wza}, makes this a very attractive model to study in the context of quantum gravity.

Recent developments in the study of the canonical tensor model sparked interest in the tensor rank decomposition from the perspective of quantum gravity. The tensor rank decomposition is a decomposition of tensors into a sum of rank-1 tensors~\cite{Hitchcock1927}, also called simple tensors, and it might be seen as a generalisation of the singular value decomposition of matrices to tensors.\footnote{For more information we would like to refer to Appendix~\ref{app:sec:TRD}} It is a tool frequently used in a broad range of sciences as it is often a very effective way to extract information from a tensor~\cite{koldaTRD}.

In~\cite{Kawano:2018pip} the tensor rank decomposition was used to extract topological and geometric information from tensors used in the CTM. Here, every term in the decomposition corresponds to a (fuzzy) point, collectively forming a space that models a universe. However, finding the exact tensor rank decomposition of a tensor is in general next to impossible~\cite{Hillar_NPhard}. This means that for a given tensor $Q_{abc}$, which is in the CTM the fundamental variable that is supposed to represent a spatial slice of spacetime, it may  potentially be approximated by several different decompositions, possibly corresponding to different universes. This leads to two questions related to the stability of this approach: 
\begin{itemize}
    \item How many tensor rank decompositions are close to a given tensor $Q_{abc}$? 
    \item Do different decompositions describe the same space (and if not; how much do they differ)?
\end{itemize}
In this work we focus on the former of these questions. To understand this question we introduce the configuration space of tensor rank decompositions for rank $R$, denoted by  $\mathcal{F}_R$, and introduce the quantity to describe the volume of the configuration space close to a tensor $Q$:\footnote{This is a formal definition which will be properly regulated later on.}
\begin{equation*}
    \mathcal{V}_R(Q, \Delta) = \int_{\mathcal{F}_R} {\rm d}\Phi \, \Theta(\Delta - \|Q-\Phi\|^2),
\end{equation*}
where $\Phi\in\mathcal{F}_R$ denotes a tensor rank decomposition in the space of tensor rank decompositions that is integrated over, $\Theta(x)\, (x\in \mathbb{R})$ is the Heaviside step function, and $\Delta$ is a parameter to define the maximum square distance between $Q$ and $\Phi$. Understanding this quantity better will lead to a better understanding of the tensor rank decomposition configuration space, and what to expect when aiming to approximate a tensor by a tensor rank decomposition.

Another motivation coming from the CTM to study the configuration space of tensor rank decompositions is coming from the quantum CTM. A noteworthy fact about the CTM is that it has several known exact solutions to the quantum constraint  equations~\cite{Narain:2014cya}. One of these has recently been extensively analysed due to the emergence of Lie-group symmetries in this wave function, which  potentially hints towards the emergence of macroscopic spacetimes~\cite{Obster:2017pdq, Obster:2017dhx, Lionni:2019rty, Sasakura:2019hql, Obster:2020vfo, Sasakura:2021lub}. This wave function, in the $Q$-representation, is closely related to a statistical model~\cite{Sasakura:2021lub} that is mathematically equivalent to
\begin{equation*}
    \Psi(Q) = \int_{\mathcal{F}_R} {\rm d}\Phi\, \mathcal{O}(\Phi) \, \e^{-\kappa (Q-\Phi)^2},
\end{equation*}
where $\mathcal{O}(\Phi)$ only depends on the weights of the components of the decomposition, which will be more precisely defined below. This shows that for a full understanding of this statistical model, understanding the underlying configuration space and the behaviour of volumes therein is important.

Besides research in the CTM, this work might be applicable more generally. Similar questions might arise in other areas of science, and mathematically there are a lot of open questions about the nature of the tensor rank decomposition. Understanding the configuration space constructed here might lead to significant insights elsewhere. For these reasons, the content of the paper is kept rather general. Our main research interests are real symmetric tensors of degree three, but we will consider both symmetric and generic (non-symmetric) tensors of general degree. 

This work is structured as follows. We define the configuration space of tensor rank decompositions in section~\ref{sec:configspace}. Here we also give a proper definition of $\mathcal{V}_R(Q,\Delta)$, and introduce the main quantity we will analyse, $\mathcal{Z}_R(\Delta)$, which is the average of $\mathcal{V}_R(Q,\Delta)$ over normalised tensors. Section~\ref{sec:derivation} contains the main result of our work. There we derive a closed formula for $\mathcal{Z}_R(\Delta)$, which is guaranteed to exist under the condition that a certain quantity $G_R$, which is independent of $\Delta$, exists and is finite. Another interesting connection to the CTM is found at this point, since this quantity $G_R$ is a generalisation of the partition function of the matrix model studied in~\cite{Lionni:2019rty, Sasakura:2019hql, Obster:2020vfo}. In section~\ref{sec:existence}, the existence of $G_R$ is proven for $R=1$, and numerical analysis is done for $R>1$ for a specific choice of volume form ${\rm d}\Phi$ to arrive at a conjecture for the maximal allowed value of $R$, called $R_c$. In section~\ref{sec:numerics} we present direct numerical computations of $\mathcal{Z}_R(\Delta)$ to further verify the analytical derivation, and conclude that the closed form indeed seems to be correct. Surprisingly, up to a divergent factor, the $\Delta$-behaviour still appears to hold for $R>R_c$. We finalise this work with some conclusions and discussions in section~\ref{sec:conclusion}.

\section{Volume in the space of tensor rank decompositions}\label{sec:configspace}
In this section we introduce the configuration space of tensor rank decompositions, and define the volume quantities we will analyse. We  consider two types of tensor spaces, namely the real symmetric tensors of degree $K$,  ${\rm Sym}^K(\mathbb{R}^N)$, and the space of generic (non-symmetric) real tensors, ${\mathbb{R}^N}^{\otimes K}$. This could be generalised even further in a relatively straightforward way, but for readability only these two cases will be discussed. First the symmetric case will be discussed, and afterwards the differences to the generic case will be pointed out. For more information about the tensor rank decomposition, see Appendix~\ref{app:sec:TRD} and references therein. 

Consider an arbitrary symmetric tensor of (symmetric) rank\footnote{Note that the usual definition of the rank of a tensor is the minimal value $R$ such that there is a solution to  equation~\eqref{eq:TRD_def}.} $R$ given by its tensor rank decomposition:
\begin{equation}
    \Phi_{a_1 \ldots a_K} = \sum_{i=1}^R \lambda_i \phi_{a_1}^i \ldots \phi_{a_K}^{i},\label{eq:TRD_def}
\end{equation}
where we choose $\phi_{a_k}^i$ to lie on the upper-hemisphere of the $N-1$-dimensional sphere, which we denote by $S^{N-1}_+$, and $\lambda_i \in \mathbb{R}$. This is mainly to remove redundancies, for later convenience and to make the generalisation easier.

The configuration space can now be defined as all of these possible configurations for a given rank $R$:
\begin{equation}\label{eq:def_config_space}
    \mathcal{F}_R := \mathbb{R}^R \times \underbrace{S^{N-1}_+ \times \ldots \times S^{N-1}_+}_{R\mathrm{\ times}} = \mathbb{R}^R \times{S^{N-1}_+}^{\times R}.
\end{equation}
Note that, while~\eqref{eq:TRD_def} links a given tensor rank decomposition in the space $\mathcal{F}_R$ to a tensor in the tensor space ${\rm Sym}^K(\mathbb{R}^N)$, our objects of interest are the tensor rank decompositions themselves.

We define an inner product on the tensor-space by, for $Q,P \in {\rm Sym}^K(\mathbb{R}^N)$,
\begin{equation}
    Q \cdot P = \sum_{a_1\ldots a_K = 1}^{N} Q_{a_1 \ldots a_K} P_{a_1 \ldots a_K},\label{eq:def_innerp}
\end{equation}
which induces a norm $\left\|Q\right\|^2 := Q \cdot Q = \sum_{a_1\ldots a_K=1}^N \left|Q_{a_1\ldots a_K}\right|^2$. 
We also use $Q^2\equiv\left\|Q\right\|^2$ for brevity. 
On the configuration space $\mathcal{F}_R$, we introduce a measure by the infinitesimal volume element
\begin{equation}
    {\rm d}\Phi_{w} = \prod_{i=1}^R |\lambda_i|^{w-1} {\rm d}\lambda_i\, {\rm d}\phi^i,\label{eq:def_volume_element}
\end{equation}
where ${\rm d}\lambda_i$ is the usual line-element of the real numbers, and ${\rm d}\phi^i$ is the usual volume element on the $N-1$-dimensional unit-sphere. $w$ (with $w\geq 1$) is introduced for generality. $w=1$ will turn out to be less singular, while $w=N$ corresponds to treating $(\lambda_i, \phi^i)$ as hyperspherical coordinates of $\mathbb{R}^N$.

In summary, for given rank $R$, we constructed a configuration space $\mathcal{F}_R$ in~\eqref{eq:def_config_space} with the infinitesimal volume element~\eqref{eq:def_volume_element}, taking inner product~\eqref{eq:def_innerp} on the tensor-space. If $R<R'$, then $\mathcal{F}_{R} \subset \mathcal{F}_{R'}$, and thus we have an increasing sequence of spaces, which limits to the whole symmetric tensor space of tensors of degree $K$:
\begin{equation*}
    \mathcal{F}_R \uparrow_{R\rightarrow \infty} {\rm Sym}^K(\mathbb{R}^N) \cong \mathbb{R}^{N_Q},
\end{equation*}
where $N_Q := \begin{pmatrix} N + K - 1 \\ K\end{pmatrix}$ counts the degrees of freedom of the tensor space.

A question one might ask is ``Given a tensor $Q$, how many tensor rank decompositions of rank $R$ approximate that tensor?''. For this, we define the following quantity
\begin{equation}
    \mathcal{V}_R^{\epsilon}(Q, \Delta) := \int_{\mathcal{F}_R} {\rm d}\Phi_{w} \ \Theta(\Delta - \| Q - \Phi \|^2)\, \e^{-\epsilon\, \sum_{i=1}^{R} \lambda_i^2} ,\label{eq:def_volume}
\end{equation}
where $\Delta$ is the maximum square distance of a tensor rank decomposition $\Phi_{a_1 \ldots a_K}$ to tensor $Q_{a_1 \ldots a_K}$, and $\epsilon$ is a (small) positive parameter. The exponential function is needed to regularise the integral, since even though $\Phi_{a_1 \ldots a_K}$ is bounded, the individual terms $\lambda_i \phi_{a_1}^i \ldots \phi_{a_K}^i$ might not be. This quantity gives an indication for how hard it will be to approximate a tensor $Q$ by a rank-$R$ tensor rank decomposition; a large value means there are many decompositions that approximate the tensor, while a small value might indicate that a larger rank is necessary. 

While~\eqref{eq:def_volume} might contain all information one would want, it is hard to compute. Instead, we will introduce a quantity to make general statements about the configuration space by averaging this quantity over all normalised tensors $\tilde{Q}_{a_1\ldots a_K}$ (such that $\|\tilde{Q}\|^2 = 1$):
\begin{equation}
    \mathcal{Z}_R(\Delta; \epsilon) := \frac{1}{V_{\|Q\|=1}} \int_{\|Q\|=1} {\rm d}\tilde{Q}\ \mathcal{V}_R^{\epsilon}(\tilde{Q}, \Delta).\label{eq:def_Z_epsilon}
\end{equation}
Since the configuration space of $Q$ is isometric to $\mathbb{R}^{N_Q}$, it is possible to move to hyperspherical variables. $\tilde{Q}$ is then given by the angular part of $Q$. Furthermore we have defined $V_{\|Q\|=1} := \int_{\|Q\|=1} {\rm d}\tilde{Q} = \frac{2 \pi ^{N_Q/2}}{\Gamma(N_Q/2)}$. For now we assume the existence of the $\epsilon\rightarrow0^+$ limit of this quantity, such that
\begin{equation}
    \mathcal{Z}_R(\Delta) := \lim_{\epsilon\rightarrow0^+} \mathcal{Z}_R(\Delta; \epsilon).\label{eq:def_Z_limit}
\end{equation}
This limit does not necessarily exist, and it diverges if $R$ is taken too large, as we will show in section~\ref{sec:existence}. In proposition~\ref{prop:Z} in the next section we will obtain an explicit formula for $\mathcal{Z}_R(\Delta)$ found in~\eqref{eq:result_Z_symmetric} under the condition that the following quantity exists:
\begin{equation}
    G_R := \lim_{\epsilon\rightarrow0^+} G_R(\epsilon) := \lim_{\epsilon\rightarrow0^+} \int_{\mathcal{F}_R} {\rm d}\Phi_{w}\ {\rm e}^{- \Phi^2 - \epsilon\, \sum_{i=1}^R \lambda_i^2 }.\label{eq:def_G}
\end{equation}
Note that, since $G_R(\epsilon)$ is a monotonically decreasing positive function of $\epsilon$, 
the $\epsilon\rightarrow 0^+$ limit either diverges or is finite if it is bounded from above. 

This condition presents a peculiar connection to the canonical tensor model. Let us first rewrite
\begin{equation}
    G_R(\epsilon) = \int_{\mathcal{F}_R} \prod_{i=1}^R {\rm d}\lambda_i |\lambda_i|^{w-1} {\rm d}\phi^i {\rm e}^{- \sum_{i,j=1}^R \lambda_i (\phi^i\cdot\phi^j)^K \lambda_j - \epsilon\, \sum_{i=1}^R \lambda_i^2 }\label{eq:G_rewritten},
\end{equation}
where we introduced the usual inner product on $S_+^{N-1} \subset \mathbb{R}^N$
\begin{equation*}
    \phi^i\cdot\phi^j = \sum_{a=1}^N \phi^i_a \phi^j_a,
\end{equation*}
inherited from the tensor space inner product. In~\cite{Lionni:2019rty,Sasakura:2019hql,Obster:2020vfo}, a matrix model was analysed that corresponds to a simplified wave function of the canonical tensor model. The matrix model under consideration had a partition function given by
\begin{equation*}
    Z(k) = \int_{\mathbb{R}^{N R}} \prod_{i=1}^R \prod_{a=1}^N {\rm d}\rho_a^i \, {\rm e}^{-\sum_{i,j=1}^R (\rho^i\cdot\rho^j)^3 - k\sum_{i=1}^R (\rho^i\cdot\rho^i)^3},
\end{equation*}
where $\rho^i \in \mathbb{R}^N$ with the usual Euclidean inner product on $\mathbb{R}^N$. Let us now go to hyperspherical coordinates $(r_i, \phi^i)$ for every $N$-dimensional subspace for every $i$, but instead of taking the usual convention where $r_i \geq 0$ and $\phi^i\in S^{N-1}$, we let $r_i\in\mathbb{R}$ and $\phi^i\in S_+^{N-1}$. Then 
\begin{align}
    Z(k) &= \int \prod_{i=1}^R |r_i|^{N-1} {\rm d}r_i {\rm d}\phi^i \, {\rm e}^{-\sum_{i,j=1}^R (r_i (\phi^i\cdot\phi^j) r_j)^3 - k \sum_{i=1}^R r_i^{6}},\nonumber\\
    &= const. \int_{\mathcal{F}_R} \prod_{i=1}^R |\lambda_i|^{\frac{N-3}{3}} {\rm d}\lambda_i {\rm d}\phi^i {\rm e}^{- \sum_{i,j=1}^R \lambda_i (\phi^i\cdot\phi^j)^3 \lambda_j - k \sum_{i=1}^R\lambda_i^2} \label{eq:Z_CTM_rewritten} ,
\end{align}
where we have substituted $\lambda_i = r_i^3$
and $const.$ is an irrelevant  numerical factor. 
Comparing~\eqref{eq:Z_CTM_rewritten} with~\eqref{eq:G_rewritten} we see that the matrix model studied in the context of the canonical tensor model is a special case of $G_R(\epsilon)$, where $\epsilon=k$, $K=3$ and $w=\frac{N}{K}$. 

Let us now turn to the case of generic (non-symmetric) tensors. We will point out the differences in the treatment and the result, though the derivation in section~\ref{sec:derivation} will be identical. We will still focus on tensors of degree $K$ that act on a multiple of Euclidean vector spaces $V=\mathbb{R}^N$, though generalisations of this could also be considered in a very similar way. A generic rank $R$ tensor is given by
\begin{equation*}
    \Phi_{a_1\ldots a_K}^{(G)} = \sum_{i=1}^{R} \lambda_i {\phi_{a_1}^{(1)}}^{i}\ldots{\phi_{a_K}^{(K)}}^{i},
\end{equation*}
where we again choose $\lambda_i\in\mathbb{R}$ and ${\phi^{(k)}}^i \in S_+^{N-1}$. Note that the main difference here is that the vectors ${\phi^{(k)}}^i$ are independent, and thus the generic configuration space will be bigger:
\begin{equation}
    \mathcal{F}_{R,K}^{(G)} := \mathbb{R}^R\times \underbrace{{S_+^{N-1}}^{\times K}\times \ldots\times {S_+^{N-1}}^{\times K}}_{R\mathrm{\ times}} = \mathbb{R}^{R}\times {S_+^{N-1}}^{\times KR},
\end{equation}
where we now define the measure by the volume element
\begin{equation}
    {\rm d}\Phi_{w}^{(G)} = \prod_{i=1}^{R}|\lambda_i|^{w-1} {\rm d}\lambda_i \prod_{k=1}^K {\rm d}{\phi^{(k)}}^i.
\end{equation}
Note that the degrees of freedom of the tensor space are now $N_Q = N^K$. Under these changes we can again define 
analogues of~\eqref{eq:def_volume},~\eqref{eq:def_Z_limit} and~\eqref{eq:def_G}. With these re-definitions, the general
result~\eqref{eq:result_Z_symmetric} will actually be the same but now for $N_Q=N^K$ and $R$ being 
the generic tensor rank (instead of the symmetric rank).

\section{Derivation of the average volume formula}\label{sec:derivation}
In this section we will derive the result as presented in~\eqref{eq:result_Z_symmetric}. The main steps of the derivation are performed in this section, but for some mathematical subtleties we will refer to appendix~\ref{app:sec:results} and for some general formulae to appendix~\ref{app:sec:formulae}. The general strategy for arriving at~\eqref{eq:result_Z_symmetric} is to take the Laplace transform, extract the dependence on the variables, and take the inverse Laplace transform.

Let us take the Laplace transform of \eqref{eq:def_Z_limit} with \eqref{eq:def_volume} and \eqref{eq:def_Z_epsilon} (see appendix~\ref{app:sec:inverse_laplace}):
\begin{align*}
    \bar{\mathcal{Z}}_R(\gamma) &= \int_{0}^\infty {\rm d}\Delta\, \mathcal{Z}_R(\Delta) \e^{-\gamma \Delta},\\
    &= \frac{1}{V_{\|Q\|=1}} \lim_{\epsilon\rightarrow0^+}\int_{\|Q\|=1} {\rm d} \tilde{Q}\, \int_{\mathcal{F}_R} {\rm d}\Phi_{w}\, \int_{\|\tilde{Q}-\Phi\|^2}^\infty {\rm d}\Delta\, \e^{-\gamma \Delta - \epsilon \sum_{i=1}^R \lambda_i^2},\\
    &= \frac{1}{\gamma V_{\|Q\|=1}} \lim_{\epsilon\rightarrow0^+} \int_{\|Q\|=1} {\rm d} \tilde{Q}\, \int_{\mathcal{F}_R} {\rm d}\Phi_{w} \, \e^{-\gamma (\tilde{Q}-\Phi)^2 - \epsilon \sum_{i=1}^R \lambda_i^2},
\end{align*}
where we have taken the limit out of the $\Delta$ integration. It will be shown below when this is allowed. Let us multiply this quantity by $\gamma$
\begin{equation}
    \bar{Z}_R(\gamma) :=  \gamma \bar{\mathcal{Z}}_R(\gamma) = \frac{1}{ V_{\|Q\|=1}}\lim_{\epsilon \rightarrow 0^+}\int_{\|Q\|=1} {\rm d}\tilde{Q}\, \int_{\mathcal{F}_R} {\rm d}\Phi_{w}\ \e^{-\gamma(\tilde{Q} - \Phi)^2  - \epsilon \sum_{i=1}^R \lambda_i^2}.\label{eq:def_Z_gamma}
\end{equation}
This will be undone again at a later stage. For later use we will also define the quantity depending on $\epsilon$ without taking the limit:
\begin{equation}
    \bar{Z}_R(\gamma; \epsilon) := \frac{1}{V_{\|Q\|=1}} \int_{\|Q\|=1} {\rm d}\tilde{Q}\, \int_{\mathcal{F}_R} {\rm d}\Phi_{w}\, \e^{-\gamma (\tilde{Q}-\Phi)^2 - \epsilon \sum_{i=1}^R \lambda_i^2}.\label{eq:def_Z_gamma_epsilon}
\end{equation}
As an aside; recall that for the Laplace transform multiplication by $\gamma$ corresponds to taking the derivative in $\Delta$-space. This means that we effectively now have a definition of the Laplace transform of the distributive quantity 
\begin{equation*}
    Z_R(\Delta; \epsilon) := \int_{\|Q\|=1} {\rm d}\tilde{Q}\  \mathcal{D}V_R^\epsilon(\tilde{Q}, \Delta) := \int_{\|Q\|=1} {\rm d}\tilde{Q} \,\int_{\mathcal{F}_R} {\rm d}\Phi_{w}\ \delta(\Delta - \| \tilde{Q} - \Phi \|^2)\, \e^{ - \epsilon \sum_{i=1}^R \lambda_i^2},
\end{equation*}
where $\delta(x)\ (x\in \mathbb{R})$ is the delta distribution, 
assuming that~\eqref{eq:def_Z_gamma} is well-defined (which will be shown below for the aforementioned assumption). 

We will now present the first main result that will be necessary. 
\begin{proposition}\label{prop:barZ_result}
Given that~\eqref{eq:def_G} is finite,~\eqref{eq:def_Z_gamma} is finite and given by
\begin{equation*}
    \bar{Z}_R(\gamma) = G_R\, \gamma^{-\frac{w\,R}{2}}{}_1F_1\left(\frac{N_Q-w\,R}{2}, \frac{N_Q}{2}, - \gamma \right).
\end{equation*}
\end{proposition}
\begin{proof}
    Let us prove this proposition in the following two steps.

    {\bf Step one}: $\bar{Z}_R(\gamma)$ is finite if $G_R$ is finite.\\
    First let us remark that the integrand in~\eqref{eq:def_Z_gamma_epsilon} is positive, and thus for $\bar{Z}_R(\gamma)$ to be finite we should show that $\bar{Z}_R(\gamma)<\infty$. Furthermore, because of the reverse triangle inequality we have the inequality
    \begin{equation*}
        \|Q - \Phi\|^2 \geq ( \|Q\| - \|\Phi\| )^2, 
    \end{equation*}
    and from $(x-y)^2 = A y^2 - \frac{A}{1-A} x^2 + (1-A)\left(y-\frac{x}{1-A}\right)^2$ for $x, y \in \mathbb{R}$ and $0 < A < 1$ we have the inequality 
    \begin{equation*}
        \left(\|Q\| - \|\Phi\|\right)^2 \geq A \|\Phi\|^2 - \frac{A}{1-A}\|Q\|^2.
    \end{equation*}
    Putting this together, we find that
    \begin{align}
        \bar{Z}_R(\gamma; \epsilon) &= \frac{1}{V_{\|Q\|=1}} \int_{\|Q\|=1} {\rm d}\tilde{Q}\, \int_{\mathcal{F}_R}{\rm d}\Phi_{w}\,\e^{-\gamma (\tilde{Q} - \Phi)^2 - \epsilon \sum_{i=1}^R \lambda_i^2},\nonumber\\
        &\leq \frac{1}{V_{\|Q\|=1}} \int_{\|Q\|=1} {\rm d}\tilde{Q}\, \int_{\mathcal{F}_R}{\rm d}\Phi_{w}\,\e^{-\gamma(\|\tilde{Q}\| - \|\Phi\|)^2 - \epsilon\sum_{i=1}^R \lambda_i^2},\nonumber\\
        &\leq \frac{1}{V_{\|Q\|=1}} \int_{\|Q\|=1} {\rm d}\tilde{Q}\, \e^{\gamma \frac{A}{1-A}\tilde{Q}^2} \int_{\mathcal{F}_R}{\rm d}\Phi_{w}\,\e^{ -\gamma A \Phi^2- \epsilon\sum_{i=1}^R \lambda_i^2},\nonumber\\
        &= (\gamma A)^{-\frac{w\, R}{2}} \e^{\gamma \frac{A}{1-A}} G_R\left(\frac{\epsilon}{\gamma A}\right).\label{eq:Z_R_bound}
    \end{align}
    This means that, as long as $G_R = \lim_{\epsilon\rightarrow0^+} G_R(\epsilon)$ is finite, $\bar{Z}_R(\gamma)=\lim_{\epsilon\rightarrow0^+} \bar{Z}_R(\gamma; \epsilon)$ is finite since we have a finite upper bound. Moreover, it converges since it monotonically increases with $\epsilon\rightarrow0^+$ and it is bounded. 
    
    {\bf Step two}: Find the closed form.\\
    Let us introduce the quantity
    \begin{equation}
        Y(\alpha, \gamma) := \lim_{\epsilon\rightarrow0^+} \int_{\mathbb{R}^{N_Q}} {\rm d}Q\, \int_{\mathcal{F}_R} {\rm d}\Phi_{w}\, \e^{-\alpha Q^2 - \gamma (Q - \Phi)^2 - \epsilon \sum_{i=1}^R \lambda_i^2}.\label{eq:def_Y_alpha_gamma}
    \end{equation}
    Note that in this quantity, $Q$ is defined over the whole tensor space $\mathbb{R}^{N_Q}$, so not only the normalised tensors. In the appendix, lemma~\ref{lemma:Y_alpha_gamma_finite} shows that this quantity is finite under the same assumption that $G_R$ is finite. 
    
    We can rewrite~\eqref{eq:def_Y_alpha_gamma} in terms of $G_R$ as follows
    \begin{align}
        Y(\alpha, \gamma) &= \lim_{\epsilon\rightarrow0^+} \int_{\mathbb{R}^{N_Q}} {\rm d}Q\, \int_{\mathcal{F}_R}{\rm d}\Phi_{w}\,\e^{-(\alpha+\gamma)\left(Q - \frac{\gamma}{\alpha+\gamma} \Phi \right)^2 - \frac{\alpha\gamma}{\alpha+\gamma} \Phi^2 -\epsilon\sum_{i=1}^R \lambda_i^2 },\nonumber\\
        &= \left(\frac{\pi}{\alpha + \gamma}\right)^{\frac{N_Q}{2}} \lim_{\epsilon\rightarrow0^+} \int_{\mathcal{F}_R} {\rm d}\Phi_w \ \e^{-\frac{\alpha \gamma}{\alpha+\gamma} \Phi^2  -\epsilon\sum_{i=1}^R \lambda_i^2},\nonumber\\
        &= \left(\frac{\pi}{\alpha + \gamma}\right)^{\frac{N_Q}{2}} \left(\frac{\alpha+\gamma}{\alpha\gamma}\right)^{\frac{w\,R}{2}} G_{R},\nonumber\\
        &= \pi^{N_Q/2} \gamma^{-\frac{N_Q + w\,R}{2}} \left(1+t\right)^{-\frac{N_Q-w\,R}{2}} t^{-\frac{w\,R}{2}} G_{R},\label{eq:Y_alpha_gamma_lhs}
    \end{align}
    where $t\equiv\frac{\alpha}{\gamma}$. We can also relate~\eqref{eq:def_Y_alpha_gamma} to $\bar{Z}_R(\gamma)$ by using polar coordinates for $Q \rightarrow (|Q|, \tilde{Q})$:
    \begin{align}
       Y(\alpha,\gamma) &= \lim_{\epsilon\rightarrow0^+} \int_{\mathbb{R}^{N_Q}} {\rm d}|Q|\, |Q|^{N_Q-1} \, {\rm d}\tilde{Q}\, \int_{\mathcal{F}_R} {\rm d}\Phi_{w}\, \e^{-\alpha |Q|^2 - \gamma (|Q| \tilde{Q} - \Phi)^2 -\epsilon\sum_{i=1}^R \lambda_i^2},\nonumber\\
    &= V_{\|Q\|=1} \lim_{\epsilon\rightarrow0^+}\int_{0}^\infty {\rm d}|Q|\, |Q|^{N_Q - 1 + w\,R} \e^{-\alpha |Q|^2} \bar{Z}_R(\gamma |Q|^2; \epsilon |Q|^2) ,\nonumber\\
    &=  \frac{1}{2} V_{\|Q\|=1} \gamma^{-\frac{N_Q + w\,R}{2}} \lim_{\epsilon\rightarrow0^+}\int_{0}^{\infty} {\rm d}x\, x^{\frac{N_Q + w\,R}{2}-1} \bar{Z}_R(x; \epsilon x/\gamma) \e^{-t\,x},\nonumber\\
    &= \frac{1}{2}V_{\|Q\|=1}\gamma^{-\frac{N_Q+w\,R}{2}} \int_{0}^{\infty} {\rm d}x\, x^{\frac{N_Q+w\,R}{2}-1} \bar{Z}_R(x) \e^{-t\,x}.\label{eq:Y_alpha_gamma_limits}
    \end{align}
    Here, in the first step we rescaled $\lambda_i\rightarrow|Q|\lambda_i$, in the second step we introduced a new integration variable $x\equiv \gamma |Q|^2$, and in the final step we took the limit inside the integral as is proven to be allowed in the appendix lemma~\ref{lemma:Y_integral_limit}. Note the appearance of $\bar{Z}_R(\gamma; \epsilon)$ as defined in~\eqref{eq:def_Z_gamma_epsilon}.
    
    By equating~\eqref{eq:Y_alpha_gamma_lhs} and~\eqref{eq:Y_alpha_gamma_limits}, we now arrive at the relation
    \begin{equation*}
        \int_{0}^{\infty} {\rm d}x \, x^{\frac{N_Q+w\,R}{2}-1} \bar{Z}_R(x) \e^{-t\,x} =  \Gamma[N_Q/2]\, G_R\, (1+t)^{-\frac{N_Q-w\,R}{2}} t^{-\frac{w\,R}{2}}.
    \end{equation*}
    The crucial observation now is that the left-hand side is the Laplace transform of the function $x^{\frac{N_Q+w\,R}{2}-1} \bar{Z}_R(x)$. Hence, by taking the inverse Laplace transform of the right-hand side and using~\eqref{eq:appendix:InverseLaplace_1} in the appendix, we find
    \begin{equation*}
        \bar{Z}_R(x) = G_R \, x^{-\frac{w\,R}{2}} {}_1F_1\left(\frac{N_Q-w\,R}{2}, \frac{N_Q}{2}, -x \right).
    \end{equation*}
\end{proof}

Having obtained the result above, we undo the operation done in~\eqref{eq:def_Z_gamma}:
\begin{equation}
    \bar{\mathcal{Z}}_R(\gamma) = G_R \gamma^{-\frac{w\,R}{2} - 1 } {}_1F_1\left(\frac{N_Q-w\,R}{2}, \frac{N_Q}{2}, -\gamma\right).\label{eq:result_bar_Z}
\end{equation}
The main remaining task to find the central result of this paper, an expression for $\mathcal{Z}_R(\Delta)$, is to take the inverse Laplace transform of this function. This is performed in the proposition below.

\begin{proposition}\label{prop:Z}
    Given that $G_R$ in~\eqref{eq:def_G} is finite, $\mathcal{Z}_R(\Delta)$, as defined in~\eqref{eq:def_Z_limit}, is given by
\begin{equation}
   \mathcal{Z}_R(\Delta) = \frac{2G_R}{\Gamma\left[\frac{w\,R}{2}\right]} \cdot \begin{cases}
    \frac{1}{N_Q}\Delta^{\frac{N_Q}{2}} {}_2F_1\left(1-\frac{w\,R}{2}, \frac{N_Q-w\,R}{2}, 1+\frac{N_Q}{2}, \Delta\right), & \Delta\leq1,\\
    \frac{1}{w\,R}\Delta^{\frac{w\,R}{2}} {}_2F_1\left(-\frac{w\,R}{2}, \frac{N_Q-w\,R}{2}, \frac{N_Q}{2}, 1/\Delta\right), & \Delta\geq1,
  \end{cases}\label{eq:result_Z_symmetric}
\end{equation}
\end{proposition}
\begin{proof}
    If~\eqref{eq:def_G} is finite, and thus~\eqref{eq:result_bar_Z} exists and is finite, we need to perform the inverse Laplace transform of~\eqref{eq:result_bar_Z} in order to prove~\eqref{eq:result_Z_symmetric}. This may be done as follows. First we write~\eqref{eq:result_bar_Z} in terms of one of the Whittaker functions
    \begin{equation*}
        \bar{\mathcal{Z}}_R(\gamma) =  G_R \, \gamma^{-\frac{w\,R}{2} - \frac{N_Q}{4} - 1} \e^{-\frac{\gamma}{2}} \, M_{\frac{N_Q}{4} - \frac{w\,R}{2} , \frac{N_Q}{4}- \frac{1}{2}}(\gamma),
    \end{equation*}
    where we used Kummer's transformation~\eqref{eq:app:kummer_theorem}, and $M_{\mu,\nu}(\gamma)$ is one of the Whittaker functions which may be found in \eqref{eq:app:confl_whittaker}
    in the appendix. Let us rewrite
    \begin{equation*}
        \bar{\mathcal{Z}}_R(\gamma) =  G_R \, \underbrace{ \gamma^{- \frac{N_Q}{4}} \e^{-\frac{\gamma}{2}} \, M_{\frac{N_Q}{4} - \frac{w\,R}{2} , \frac{N_Q}{4}- \frac{1}{2}}(\gamma)}_{L[f]} \, \underbrace{\gamma^{-\frac{w\,R}{2} - 1}}_{L[g]},
    \end{equation*}
    such that we can now use the formula from the convolution theorem which can be found in \eqref{eq:appendix:ProductFormula}
    in the appendix. Let us first find the inverse Laplace transform of $L[g]$, which may be found using formula~\eqref{eq:appendix:InverseLaplace_2} from the appendix
    \begin{equation*}
        g(t) = \frac{t^{\frac{w\,R}{2}}}{\Gamma\left[\frac{w\,R}{2} + 1\right]}.
    \end{equation*}
    The inverse Laplace transform of $L[f]$ may be found using formula~\eqref{eq:appendix:InverseLaplace_3} from the appendix
    \begin{equation*}
        f(t) = \begin{cases}
            \beta\left(\frac{w\,R}{2}, \frac{N_Q-w\,R}{2}\right)^{-1} t^{\frac{N_Q-w\,R}{2} - 1} (1-t)^{\frac{w\,R}{2}-1}, & 0 < t < 1, \\
            0 , & \text{otherwise} ,
        \end{cases}
    \end{equation*}
    where $\beta$ is the beta-function defined in~\eqref{eq:app:beta_f}. Combining these results with the convolution product formula~\eqref{eq:appendix:ProductFormula} in the appendix yields
    \begin{equation*}
        \mathcal{Z}_R(\Delta) = \begin{cases}
           c_{R} \int_{0}^{\Delta} q^{\frac{N_Q-w\,R}{2}-1} (1-q)^{\frac{w\,R}{2}-1} (\Delta-q)^{\frac{w\,R}{2}}\, {\rm d}q , & \Delta\leq1, \\
           c_{R} \int_{0}^{1} q^{\frac{N_Q-w\,R}{2}-1} (1-q)^{\frac{w\,R}{2}-1} (\Delta-q)^{\frac{w\,R}{2}}\, {\rm d}q , & \Delta\geq1 ,
        \end{cases}
    \end{equation*}
    where $c_{R} \equiv \frac{G_R}{\Gamma\left[\frac{w\,R}{2}+1\right] \beta\left(\frac{w\,R}{2}, \frac{N_Q-w\,R}{2}\right)}$. Let us focus on the $\Delta\geq1$ case first. Using~\eqref{eq:appendix:hypergeometric_integral} we find
    \begin{align*}
        \mathcal{Z}_R(\Delta) &= c_{R}\, \Delta^{\frac{w\,R}{2}} \int_{0}^1 q^{\frac{N_Q-w\,R}{2}-1} (1-q)^{\frac{w\,R}{2}-1} (1-q/\Delta)^{\frac{w\,R}{2}} {\rm d}q,\nonumber\\
        &= \frac{G_R}{\Gamma\left[\frac{w\,R}{2}+1\right]} \Delta^{\frac{w\,R}{2}} {}_2F_1\left( -\frac{w\,R}{2}, \frac{N_Q-w\,R}{2}, \frac{N_Q}{2}, \frac{1}{\Delta} \right).
    \end{align*}
    For $\Delta\leq 1$ we find
    \begin{align*}
        \mathcal{Z}_R(\Delta) &= c_{R} \int_{0}^\Delta q^{\frac{N_Q-w\,R}{2}-1} (1-q)^{\frac{w\,R}{2}-1} (\Delta - q)^{\frac{w\,R}{2}}\, {\rm d}q,\nonumber\\
        &= c_{R}\, \Delta^{\frac{N_Q}{2}} \int_{0}^1 q^{\frac{N_Q-w\,R}{2}-1} (1-\Delta q)^{\frac{w\,R}{2}-1} (1 - q)^{\frac{w\,R}{2}}\, {\rm d}q,\nonumber\\
        &= \frac{G_R}{\Gamma\left[\frac{w\,R}{2}\right] \, \frac{N_Q}{2}} \Delta^{\frac{N_Q}{2}} {}_2F_1\left( 1-\frac{w\,R}{2}, \frac{N_Q-w\,R}{2}, \frac{N_Q}{2}+1, \Delta \right).
    \end{align*}
    where we changed integration variables in the first step to $q' = q/\Delta$. This result is in accord with~\eqref{eq:result_Z_symmetric}.
\end{proof}

This concludes the proof of~\eqref{eq:result_Z_symmetric}. As mentioned before, for generic tensors the derivation is exactly identical. The main difference now is that the number of degrees of freedom $N_Q$ is different for this tensor space. What are left are to determine the range of $R$ for which $G_R$ is finite and the value of $G_R$. This will be done in section~\ref{sec:existence}.

Before we finish this section, let us demonstrate some properties of this function. First let us note that the parameters $R$ and $w$ always come together, even though they seemingly are unrelated when inspecting~\eqref{eq:def_volume}. This can be understood by the fact that every term in the tensor rank decomposition comes with a weight given by $\lambda_i$. However, in the measure we count every unit of $\lambda$ with a power of $w$, so we have $R$ terms that each scale with a factor of $w$, explaining why $R$ and $w$ always come together. 

Now we take a look at some special values of the function. Starting with the case where $w\,R/2 = 1$, we have the situation that, for $\Delta \leq 1$ the hypergeometric part of the function will be constant because the first argument is zero. For $\Delta \geq 1$, we see that 
the function will be of the form $1+\frac{N_Q}{2}(\Delta-1)$. So the full function will simplify to
\begin{equation*}
    \mathcal{Z}_R(\Delta) \propto \begin{cases}
           \Delta^{N_Q/2}, & \Delta\leq 1, \\
           1+\frac{N_Q}{2}(\Delta-1) , & \Delta\geq1 ,
        \end{cases}
\end{equation*}
making the function linear for larger $\Delta$. Let us try another simple case, namely for $w\,R = N_Q$. In this case, the hypergeometric part becomes a constant everywhere, and we get
\begin{equation*}
    \mathcal{Z}_R(\Delta) \propto \Delta^{N_Q/2}.
\end{equation*}
Examples of the special values above, and others, are plotted in figure~\ref{fig:functionalFormZ}.

\begin{figure}
    \centering
    \begin{minipage}{0.48\textwidth}
        \centering
        \includegraphics[width=0.95\textwidth]{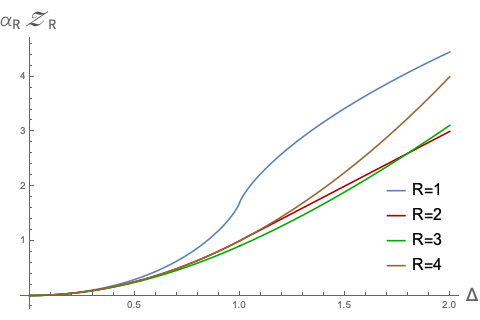}
    \end{minipage}
    \begin{minipage}{0.48\textwidth}
        \centering
        \includegraphics[width=0.95\textwidth]{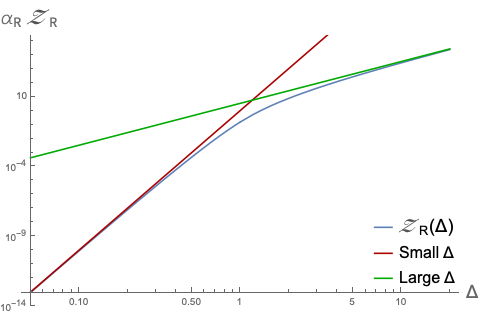}
    \end{minipage}
    \caption{On the left: $\mathcal{Z}_R(\Delta)$ for symmetric tensors with $\Delta$ running from $\Delta=0$ to $2$, where $K=3$, $N=2$ and $w=1$. 
    On the right: The limiting behaviour of $\mathcal{Z}_R(\Delta)$ for $K=3, N=4, w=2, R=3$, again for symmetric tensors. The blue curve represents~\eqref{eq:result_Z_symmetric}, the red line the small $\Delta$ behaviour of~\eqref{eq:limit_small_delta}, and the green line the large $\Delta$ behaviour of~\eqref{eq:limit_large_delta}. $\alpha_R\equiv \frac{\Gamma\left[\frac{w\,R}{2}\right]N_Q}{2 G_R}$ is a normalisation factor.}
    \label{fig:functionalFormZ}
\end{figure}
Furthermore, let us focus on some of the limiting behaviour of the function. For $\Delta\rightarrow0^+$, the hypergeometric part is approximately a constant, and we see
\begin{equation}\label{eq:limit_small_delta}
    \lim_{\Delta\rightarrow0^+}\mathcal{Z}_R(\Delta) \propto \Delta^{\frac{N_Q}{2}}.
\end{equation}
Similarly, for $\Delta\rightarrow\infty$, the hypergeometric part is constant and the function tends to
\begin{equation}\label{eq:limit_large_delta}
    \lim_{\Delta\rightarrow\infty} \mathcal{Z}_R(\Delta) \propto \Delta^{\frac{w\,R}{2}}.
\end{equation}
In some sense, the hypergeometric part of the function interpolates between these two extremes. This is also shown in figure~\ref{fig:functionalFormZ}.

\begin{figure}
    \centering
    \begin{minipage}{0.47\textwidth}
        \includegraphics[width=0.95\textwidth]{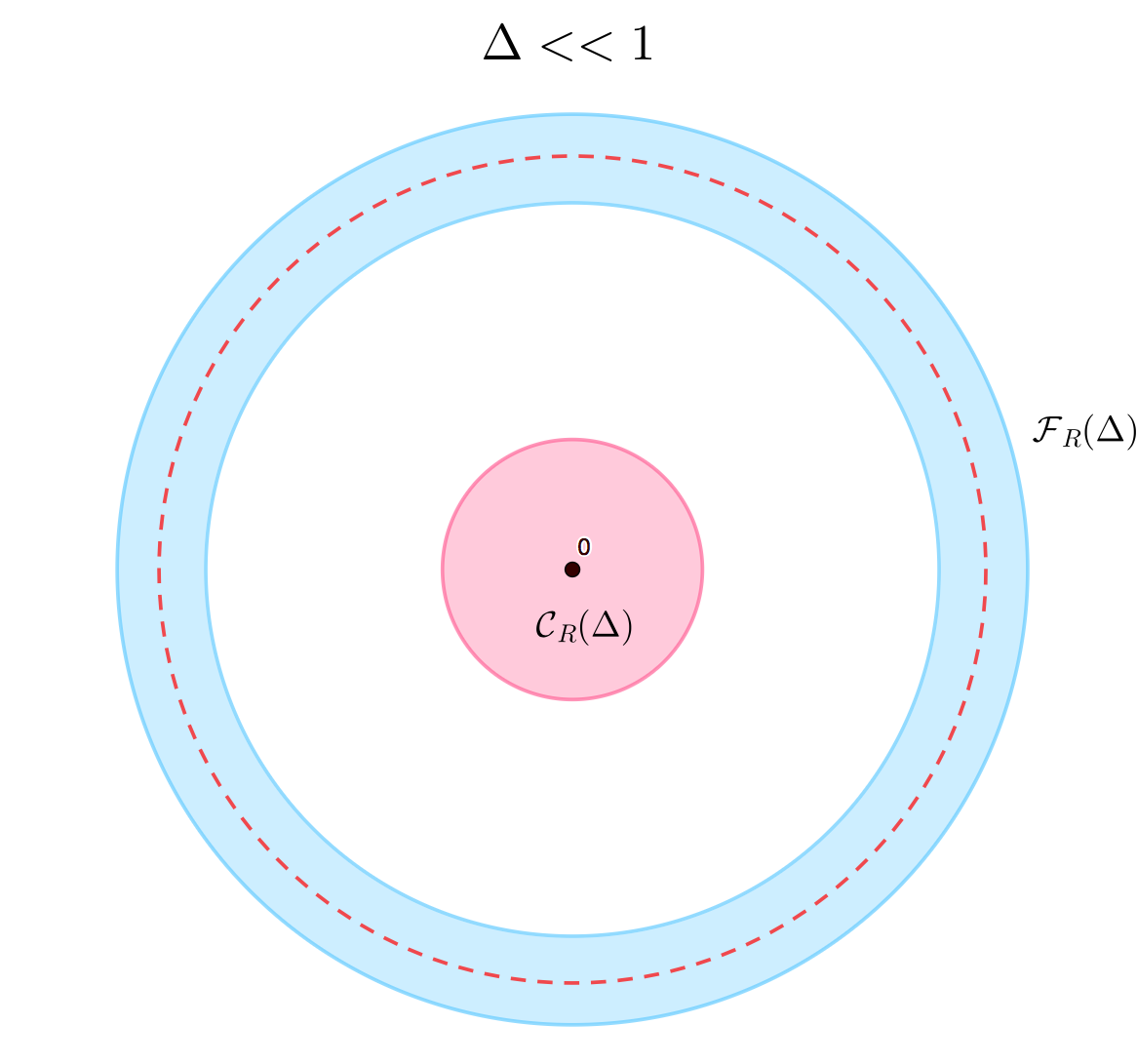}
    \end{minipage}
    \begin{minipage}{0.46\textwidth}
        \includegraphics[width=0.95\textwidth]{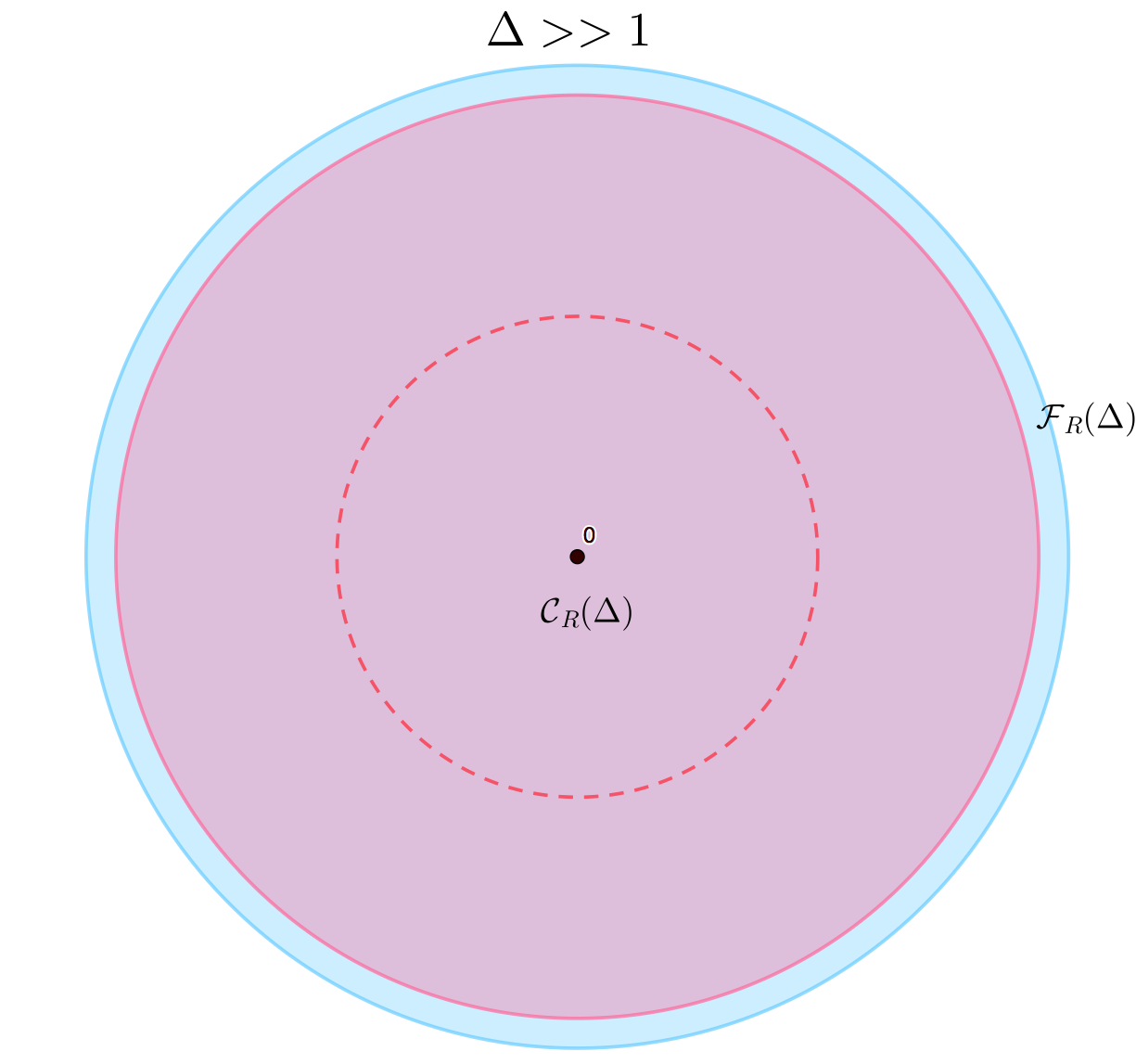}
    \end{minipage}
    
    \caption{A sketch shows the difference in the quantities $\mathcal{Z}_R(\Delta)$ and $\mathcal{C}_R(\Delta)$. The red dotted line represents the normalised tensors. The blue shaded area represents the area counted by $\mathcal{Z}_R(\Delta)$, and the red shaded area represents the area counted by $\mathcal{C}_R(\Delta)$. On the left we take $\Delta \ll 1$, and on the right we take $\Delta \gg 1$.}
    \label{fig:Z_C_sketch}
\end{figure}
It is instructive to compare $\mathcal{Z}_R(\Delta)$ to another quantity, 
\begin{align}
    \mathcal{C}_R(\Delta) &:= \int_{\mathcal{F}_R} {\rm d}\Phi_{w}\, \Theta\left( \Delta - \|\Phi\|^2 \right),\nonumber\\
    &= \frac{G_R}{\Gamma\left[\frac{w\,R}{2}+1\right]} \Delta^{\frac{w\,R}{2}}.
    \label{eq:defofcr}
\end{align}
For the derivation of this quantity we would like to refer to appendix~\ref{sec:app:delta_dep}. This quantity measures the amount of tensor rank decompositions of size smaller than $\Delta$, giving us a measure for the scaling of volume in the space of tensor rank decompositions. Figure~\ref{fig:Z_C_sketch} sketches the difference between $\mathcal{Z}_R(\Delta)$ and $\mathcal{C}_R(\Delta)$. It can be seen that in the $\Delta\rightarrow\infty$ limit, $\mathcal{Z}_R(\Delta) \rightarrow \mathcal{C}_R(\Delta)$.

Dividing $\mathcal{Z}_R(\Delta)$ by this quantity yields a quantity comparing the amount of tensor rank decompositions with a distance 
less than $\sqrt{\Delta}$ from a tensor of size 1, to the amount of decompositions of size less than $\sqrt{\Delta}$: 
\begin{equation}\label{eq:def_ZR_CR}
    \mathcal{Z}_R(\Delta)/\mathcal{C}_R(\Delta) = \begin{cases}
        \frac{w\,R}{N_Q} \Delta^{\frac{N_Q-w\,R}{2}} {}_2F_1\left(1-\frac{w\,R}{2}, \frac{N_Q-w\,R}{2}, \frac{N_Q}{2}+1, \Delta\right), &  \Delta \leq 1,\\
        {}_2F_1\left(-\frac{w\,R}{2}, \frac{N_Q-w\,R}{2}, \frac{N_Q}{2}, \frac{1}{\Delta}\right) , & \Delta \geq 1.
    \end{cases}
\end{equation}
This quantity is useful to predict the difficulty of finding a tensor rank decomposition close to a certain tensor in the tensor space. Notice here that the $G_R$ dependence drops out. This implies that this quantity might be well-defined even in the case that $G_R$ itself is not.
\begin{figure}
    \centering
    \includegraphics[width=0.67\textwidth]{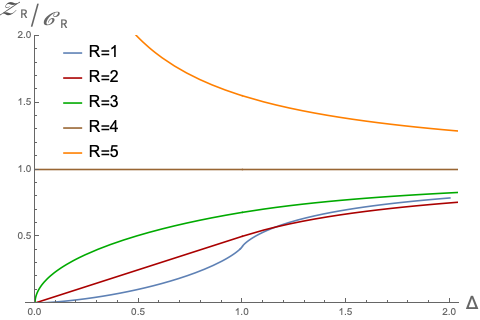}
    \caption{The quantity $\mathcal{Z}_R(\Delta)/\mathcal{C}_R(\Delta)$ for $K=3, N=2, w=1$ and $R$ ranging from 1 to 5. We can identify some of the behaviour expected from~\eqref{eq:def_ZR_CR} and~\eqref{eq:limit_large_delta}. For any value of $R$, the function nears $1$ for $\Delta\rightarrow\infty$. For $w\,R=N_Q$ the function is just one everywhere.}
    \label{fig:Z_C}
\end{figure}

Upon inspecting figure~\ref{fig:Z_C}, it can be seen that~\eqref{eq:def_ZR_CR} has some interesting $R$-dependence. Firstly, 
while the limiting behaviour for $\Delta\rightarrow\infty$ to $1$ is already clear from~\eqref{eq:limit_large_delta} and the overlap in the 
regions as sketched in figure~\ref{fig:Z_C_sketch}, the quantity will limit to $1$ from below for $w\,R<N_Q$, while for $w\,R>N_Q$ it will limit 
towards 1 from above. The reason for this is that for large $R$, even with small $\Delta$ there will be many tensor rank decompositions that 
approximate an arbitrary tensor with error allowance less than $\Delta$, while for small $\Delta$ the volume counted 
by $\mathcal{C}_R(\Delta)$ will be small. This shows that for small $\Delta$, the regions in figure~\ref{fig:Z_C_sketch} scale 
in different ways. Secondly, what is interesting is that the $R=1$ curve overtakes the $R=2$ curve around $\Delta=1$, and for larger $R$ 
the behaviour for small $\Delta$ changes from accelerating to decelerating.

This motivates us to look at a specific case of the quantity~\eqref{eq:def_ZR_CR}, namely for $\Delta=1$. As is clear from the structure of the function, $\Delta = 1$ appears to be a special value which we can analyse further. Fixing $\Delta=1$ gives us the opportunity to look at the $R$ and $w$-dependence a bit closer. Up until now we have kept the value of $w$ arbitrary, it is however interesting to see what happens for specific values of $w$. It turns out that, peculiarly, when taking 
\begin{equation}\label{eq:approximate_Q}
    w \approx \frac{K}{3}\left(N - \frac{11}{12}\right),
\end{equation}
for generic tensors, the function $\mathcal{Z}_R/\mathcal{C}_R(\Delta=1)$, as a function of $R$, appears to be minimised at (or very close to) the expected generic rank of the tensor space.\footnote{The expected rank of a tensor space is the expected rank for which $\mathcal{F}_R$ becomes dense in (an open subset of) the full tensor space. See appendix~\ref{app:sec:TRD}.} This means that until the expected rank, the relative amount of decompositions that approximate tensors is decreasing, while from the expected rank the amount of decompositions that approximate a tensor of unit norm increase. The reason for the form of~\eqref{eq:approximate_Q} is currently unknown, and it would be interesting to find a theoretical explanation for this.

\begin{figure}
    \centering
    \begin{minipage}{0.475\textwidth}
        \centering
        \includegraphics[width=0.95\textwidth]{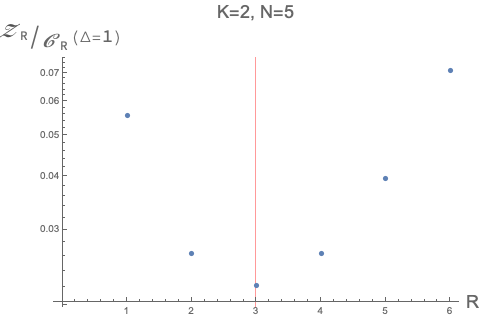}
    \end{minipage}
    \begin{minipage}{0.475\textwidth}
        \centering
        \includegraphics[width=0.95\textwidth]{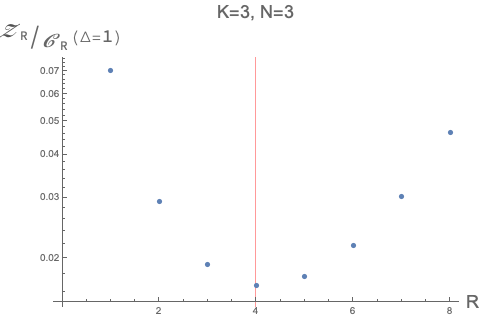}
    \end{minipage}
    \begin{minipage}{0.475\textwidth}
        \centering
        \includegraphics[width=0.95\textwidth]{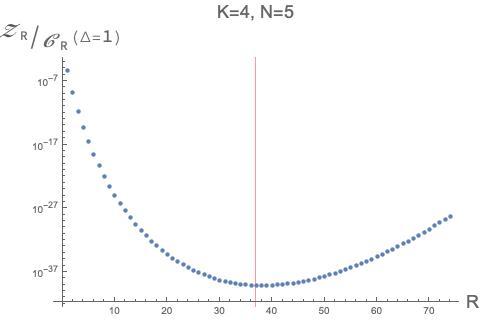}
    \end{minipage}
    \begin{minipage}{0.475\textwidth}
        \centering
        \includegraphics[width=0.95\textwidth]{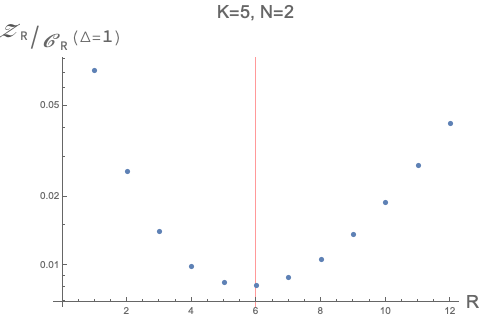}
    \end{minipage}
    \caption{Examples of the minimums when choosing $w$ to be~\eqref{eq:approximate_Q}. The horizontal axis labels $R$, while the 
    vertical axis labels $\mathcal{Z}_R/\mathcal{C}_R(\Delta=1)$. The red line represents the expected rank, see~\eqref{eq:app:expected_rank}, of the tensor space (which is 
    taken to be generic).}
    \label{fig:minimizationExpectedRank}
\end{figure}

\section{Convergence and existence of the volume formula}\label{sec:existence}
The derivation of the closed form of $\mathcal{Z}_R(\Delta)$ depends on the existence of $G_R$, defined in~\eqref{eq:def_G}. 
We will analyse the existence in the current section. 
Except for the case where $R=1$, which is shown below, we will focus on numerical results since a rigid analytic understanding is not present at this point.

First, let us briefly focus on the case of general $N, K$ and $w$, but specifically for $R=1$. This case is the only known case for general $N, K$ and $w$ that can be solved exactly. In this case the quantity simplifies to
\begin{equation*}
    G_1(\epsilon) = \int_{-\infty}^{\infty} |\lambda|^{w-1} {\rm d}\lambda \int_{S_+^{N-1}} {\rm d}\phi\, \e^{-(1+\epsilon)\lambda^2} = \frac{\Gamma\left[\frac{w}{2}\right]}{(1+\epsilon)^{w/2}}\frac{\pi^{\frac{N}{2}}}{\Gamma\left[\frac{N}{2}\right]}.
\end{equation*}
Clearly, in this case the $\lim_{\epsilon\rightarrow0^+}G_1(\epsilon)$ exists, so there exist at least one $R$ for which the quantity exists. The main question is now up to what value of $R$, $R_c$, the quantity exists.

Contrary to the $R=1$ case above, one might expect~\eqref{eq:def_G} does not always converge. The matrix model analysed in~\cite{Lionni:2019rty, Sasakura:2019hql, Obster:2020vfo}, corresponding to a choice of parameters of $K=3$ and $w=\frac{N}{K}$, did not converge in general. It had a critical value around $R_c\sim \frac{1}{2}(N+1)(N+2)$, above which the $\epsilon\rightarrow0^+$ limit did not appear to converge anymore. In the current section we will add numerical analysis for general $K$ and $w=1$, and discuss the apparent leading order behaviour. The main result of this section is that, for $w=1$, the critical value seems to be $R_c=N_Q$. Hereafter in this section we will always assume $w=1$.

The numerical analysis was done by first integrating out the $\lambda_i$ variables, and subsequently using Monte Carlo sampling on the compact manifold that remains. The derivation below is for the symmetric case, but the generic case can be done in a similar manner. The $\lambda_i$ can be integrated out in a relatively straightforward way since the measure in the $w=1$ case is very simple. Let us rewrite~\eqref{eq:def_G} in a somewhat more suggestive form
\begin{align}
    G_R(\epsilon) &:= \int_{\mathcal{F}_R} {\rm d}\Phi_{w} \, \e^{-\Phi^2 - \epsilon \sum_{i=1}^R \lambda_i^2},\nonumber\\
    &= \int_{\mathcal{F}_R} \prod_{i=1}^R {\rm d}\lambda_i\,{\rm d}\phi_i \, \e^{- \sum_{i,j=1}^R \lambda_i \left(\left(\phi^i \cdot \phi^j
    \right)^K+ \epsilon \delta^{ij}\right)\lambda_j }.
\end{align}
It can now be seen that, for $\lambda_i$, this is a simple Gaussian matrix integral over the real numbers $\lambda_i$, with the matrix $M^{ij}_{\epsilon} := \left(\phi^i\cdot\phi^j\right)^K + \epsilon \delta^{ij}$. The result of this integral is
\begin{equation*}
    G_R(\epsilon) = \left(\pi\right)^{R/2}\int_{{S_+^{N-1}}^{\times R}} \prod_{i=1}^R {\rm d}\phi^i\, \frac{1}{\sqrt{\det\left[\left(\phi^i\cdot\phi^j\right)^K + \epsilon \delta^{ij}\right]}},
\end{equation*}
which is a compact, finite (for $\epsilon>0$) integral. The corresponding expression for generic tensors is
\begin{equation*}
    G_R'(\epsilon) = \left(\pi\right)^{R/2}\int_{{S_+^{N-1}}^{\times KR}}\prod_{i=1}^R \prod_{k=1}^K {\rm d}{\phi^{(k)}}^{i} \, \frac{1}{\sqrt{\det \left[ \prod_{k=1}^K{\phi^{(k)}}^i\cdot{\phi^{(k)}}^j + \epsilon\delta^{ij} \right]}}.
\end{equation*}

We wrote a C++ program evaluating the integrals above using Monte Carlo sampling. The general method applied is the following:
\begin{enumerate}
    \item Construct $R$, $N$-dimensional random normalised vectors using Gaussian sampling.
    \item Generate the matrix $M^{ij}$ by taking inner products (and adding $\epsilon$ to the diagonal elements).
    \item Calculate the determinant of $M^{ij}$ and evaluate the integrand.
    \item Repeat this process $M$ times.
\end{enumerate}
The main difference between the above method, and the method for generic tensors, is that we generate $R\cdot K$ random vectors and the matrix is now given by $M^{ij}_\epsilon:=\prod_{k=1}^K {\phi^{(k)}}^i\cdot{\phi^{(k)}}^j + \epsilon\delta^{ij} $. To generate random numbers we used C++'s Mersenne Twister implementation mt19937, and for the calculation of the determinant of $M^{ij}_\epsilon$ we used the C++ Eigen package~\cite{eigenweb}.
\begin{figure}
    \centering
    \begin{subfigure}{0.475\textwidth}
        \centering
        \includegraphics[width=0.975\textwidth]{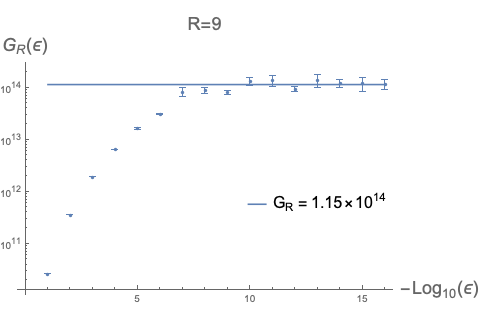}
    \end{subfigure}
    \begin{subfigure}{0.475\textwidth}
        \centering
        \includegraphics[width=0.975\textwidth]{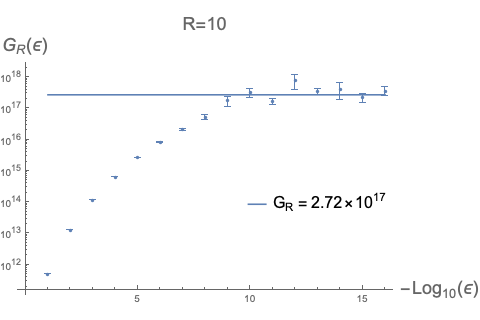}
    \end{subfigure}
    \begin{subfigure}{0.475\textwidth}
        \centering
        \includegraphics[width=0.975\textwidth]{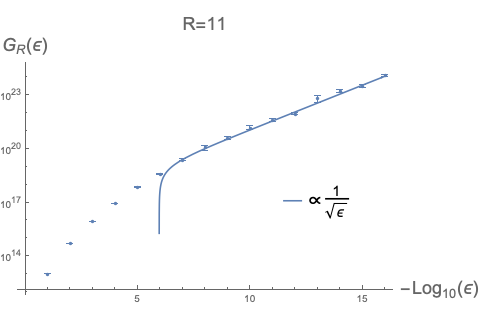}
    \end{subfigure}
    \begin{subfigure}{0.475\textwidth}
        \centering
        \includegraphics[width=0.975\textwidth]{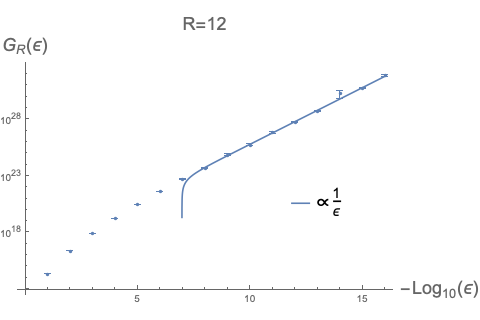}
    \end{subfigure}
    \caption{An example of the verification of $R_c$ and the determination of the numerical value of $G_R$. This is the case for symmetric tensors, with $K=3$ and $N=3$. The dots (with error-bars) represent the measurements, and the fitted curves are $C*\epsilon^{-\frac{R-R_c}{2}}+const.$ for $R>R_c$ as in~\eqref{eq:G_R_epsilon_expectation}, and the constant value $G_R$ for $R\leq R_c$ as in~\eqref{eq:def_G}. This clearly shows that in this case $R_c=10$.}
    \label{fig:evaluate_GR}
\end{figure}

We have done simulations using this method for both symmetric and generic tensors. After the initial results it became clear that the critical value for $R$ seems to lie on $R_c = N_Q$, so to verify this we calculated the integral for $R_c - 1$, $R_c$ and $R_c+1$, and checked if 
$G_R$ indeed starts to diverge at $R_c+1$. 

What divergent behaviour to expect can be explained as follows. Let us take the limit of $\lim_{\epsilon\rightarrow0^+} M^{ij}_\epsilon =: M^{ij}$. It is clear that this integral diverges whenever the matrix is degenerate. Assume now that $M^{ij}$ has rank $r$, meaning that the matrix $M^{ij}$ in diagonalised form has $R-r$ zero-entries. Thus, adding a small but positive $\epsilon$ to the diagonal entries results in the following expansion
\begin{equation*}
    \det M_\epsilon = A\,\epsilon^{R-r} + \mathcal{O}(\epsilon^{R-r-1}),
\end{equation*}
leading to leading order for the integrand
\begin{equation*}
    \frac{1}{\sqrt{\det M_\epsilon}} \sim \epsilon^{-\frac{R-r}{2}}.
\end{equation*}
Thus, if there is a set with measure nonzero in the integration region with $r<R$, the final $\epsilon$-dependence for small epsilon is expected to be
\begin{equation}
    G_R(\epsilon) \approx C\,\epsilon^{-\frac{R-R_c}{2}} + \mathcal{O}(\epsilon^{-\frac{R-R_c-1}{2}})\label{eq:G_R_epsilon_expectation}
\end{equation}
where the constant factor $C$ is the measure of the divergent set, and the other factor is due to non-leading order non-zero measure integration regions. Note that now we should take $r=R_c$, as by definition of $R_c$ this will yield the leading order contribution for the integral. An example of this approach for finding $R_c$ for symmetric tensors with $N=3$ and $K=3$ is given in figure~\ref{fig:evaluate_GR}. By the definition of $R_c$, for $R\leq R_c$, $G_R(\epsilon)$ should converge to a constant value.

This procedure has been done for both symmetric and generic tensors, and for various choices of the parameters $K$ and $N$. The results of this can be found in table~\ref{table:Rc_results}. This procedure lets us also determine the value of $G_R$ numerically, as is also shown in the examples of figure~\ref{fig:evaluate_GR}.
\begin{table}
    \begin{subtable}{0.5\textwidth}
        \centering
            \begin{tabular}{|c|c|c|c|}
            \hline
            \multicolumn{4}{|c|}{\textbf{Symmetric tensors}}\\
            \hline
                \ $K$\   & \ $N$\  & \ $R_c$\ & $N_Q$\\
                \hline
                \multirow{4}{1.1em}{\ 2} & 2 & 1 & 3 \\
                 & 3 & 6 & 6 \\
                 & 4 & 10 & 10 \\
                 & 5 & 15 & 15 \\
                \hline
                \multirow{4}{1.5em}{\ 3} & 2 & 1 & 4 \\
                 & 3 & 10 & 10 \\
                 & 4 & 20 & 20 \\
                 & 5 & 35 & 35 \\
                \hline
                \multirow{3}{1.5em}{\ 4} & 2 & 1 & 5 \\
                 & 3 & 15 & 15 \\
                 & 4 & 35 & 35 \\
                \hline
            \end{tabular}
    \end{subtable}
    \begin{subtable}{0.5\textwidth}
        \centering
            \begin{tabular}{|c|c|c|c|}
            \hline
            \multicolumn{4}{|c|}{\textbf{Generic tensors}}\\
            \hline
                \ $K$\   & \ $N$\  & \ $R_c$\ & $N_Q$\\
                \hline
                \multirow{4}{1.1em}{\ 2} & 2 & 4 & 4 \\
                 & 3 & 9 & 9 \\
                 & 4 & 16 & 16 \\
                 & 5 & 25 & 25 \\
                \hline
                \multirow{3}{1.5em}{\ 3} & 2 & 8 & 8 \\
                 & 3 & 27 & 27 \\
                 & 4 & 64 & 64 \\
                \hline
                \multirow{2}{1.5em}{\ 4} & 2 & 16 & 16 \\
                 & 3 & 81 & 81 \\
                \hline
            \end{tabular}
    \end{subtable}

    \caption{The results of the verification of $R_c$ for both symmetric tensors and generic tensors. It can be seen that in most cases, except $N=2$ for symmetric tensors, the hypothesis $R_c=N_Q$ holds.}
    \label{table:Rc_results}
\end{table}

Generally, the result was quite clear: There is a transition point at $R_c=N_Q$. This is true for all examples we tried, except for the $N=2$ cases for symmetric tensors, in which cases the critical value is $R_c=1$.

Let us explain why an upper bound for the value of $R_c$ is given by $N_Q$. The matrix may be written as
\begin{equation*}
    M^{ij} = \sum_{a_1,\ldots,a_K=1}^N (\phi_{a_1}^i \ldots \phi_{a_K}^i)\cdot(\phi_{a_1}^j \ldots \phi_{a_K}^j).
\end{equation*}
Thus, if we consider only the right part of the expression above (i.e. one of the rows of the matrix), it can be seen as the linear map
\begin{align*}
    \Lambda : \mathbb{R}^{R} &\rightarrow \mathbb{R}^{N_Q},\\
    \lambda_i &\rightarrow \sum_{i=1}^R \lambda_i \phi^i_{a_1}\ldots \phi^i_{a_K}.
\end{align*}
A basic result from linear algebra is that a linear map from a vectorspace $V$ to $W$, with $\dim(V) \geq \dim(W)$, has a kernel of at least dimension
\begin{equation*}
    \dim(\ker \Lambda) \geq \dim(V) - \dim(W).
\end{equation*}
Thus, for $R>N_Q$ this kernel always has a finite dimension, and since $M^{ij}$ is simply the square of this linear transformation, $\det M = 0$. Thus we may conclude
\begin{equation*}
    R_c \leq N_Q.
\end{equation*}
The reason why the critical rank actually attains this maximal value for all cases $N>2$ is at present not clear. However, it is good to note that for random matrices the set of singular matrices has measure zero, hence for $R\leq R_c$ the construction of the matrix $M^{ij}$ appears to be random.

The current result of $R_c=N_Q$, together with the previous result for $w=\frac{N}{K}$ and $K=3$ of $R_c \approx \frac{3}{N} N_Q$ mentioned before, suggest a general formula that holds for most cases
\begin{equation}\label{eq:conjecture:RC}
    R_c = \frac{N_Q}{w}.
\end{equation}
This formula seems very simple, but there is no analytic understanding for this formula yet. At present it should be treated merely as a conjecture.

\section{Numerical evaluation and comparison}\label{sec:numerics}
The main goal of this section is to numerically confirm the derived formula for $\mathcal{Z}_R(\Delta)$ in~\eqref{eq:result_Z_symmetric}. Therefore we will mainly focus on values of $R \leq R_c$ found in section~\ref{sec:existence} that allow for the existence of $G_R$ defined in~\eqref{eq:def_G}, since in those cases the derivation is expected to hold. We will briefly comment on cases where $R > R_c$ at the end of the section. In short; we will find that the relation found in~\eqref{eq:result_Z_symmetric} indeed holds for all cases that could reliably be calculated. In this section we will always take $w=1$, such that the integration measure on $\mathcal{F}_R$ is given by
\begin{equation*}
    {\rm d}\Phi := \prod_{i=1}^R {\rm d}\lambda_i {\rm d}\phi^i.
\end{equation*}

\begin{figure}
    \centering
    \begin{subfigure}{0.475\textwidth}
        \centering
        \includegraphics[width=0.95\textwidth]{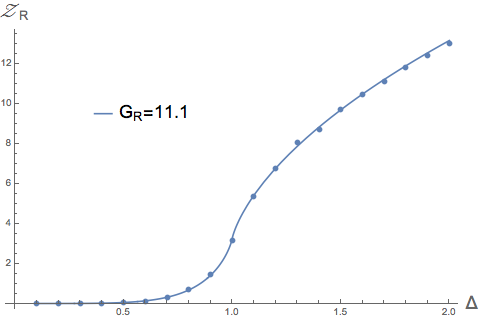}
        \caption{Symmetric tensors $N=3$, $R=1$}
    \end{subfigure}
    \begin{subfigure}{0.475\textwidth}
        \centering
        \includegraphics[width=0.95\textwidth]{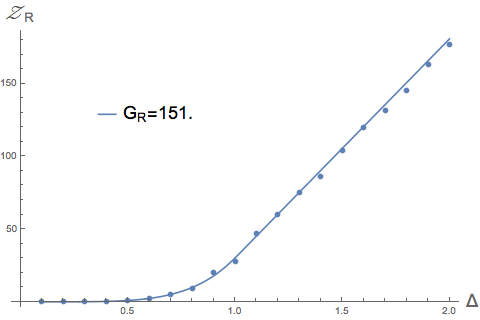}
        \caption{Symmetric tensors $N=3$, $R=2$}
    \end{subfigure}
    \begin{subfigure}{0.475\textwidth}
        \centering
        \includegraphics[width=0.95\textwidth]{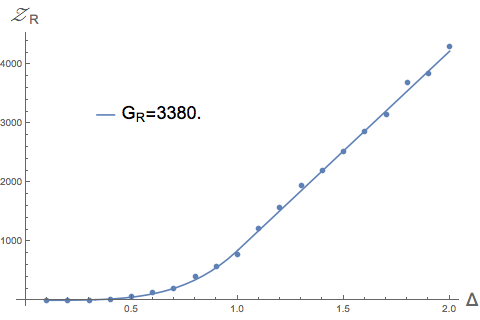}
        \caption{Generic tensors $N=2$, $R=2$}
    \end{subfigure}
    \begin{subfigure}{0.475\textwidth}
        \centering
        \includegraphics[width=0.95\textwidth]{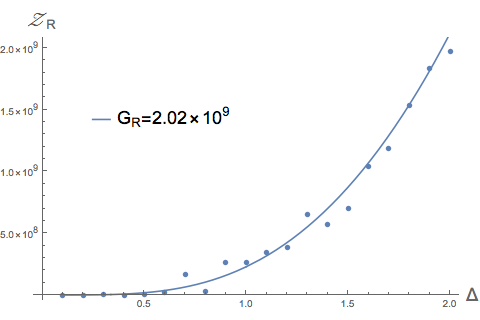}
        \caption{Generic tensors $N=2$, $R=5$}
    \end{subfigure}
    \caption{Several examples of the direct numerical evaluation of $\mathcal{Z}_R(\Delta)$ for $K=3$ and $w=1$, as a function of $\Delta$. The dots illustrate the numerically evaluated values, while the line is the curve in~\eqref{eq:result_Z_symmetric} with the value of $G_R$ determined numerically as explained in section~\ref{sec:existence}.}
    \label{fig:Z_theta_symmetric}
\end{figure}

Since the integration region has a rapidly increasing dimension, we used Monte Carlo sampling to evaluate the integral. To do this, we alter the configuration space to a compact manifold by introducing a cutoff $\Lambda$
\begin{equation*}
    \mathbb{R}^R\times {S_+^{N-1}}^{\times R} \rightarrow \left[-\Lambda, \Lambda\right]^R \times {S_+^{N-1}}^{\times R},
\end{equation*}
and similarly for the generic tensor case: 
\begin{equation*}
    \left[-\Lambda, \Lambda\right]^R \times {S_+^{N-1}}^{\times KR}.
\end{equation*}

With the integration region now being compact, there is no need for the extra regularisation parameter $\epsilon$ anymore, and we can let $\Lambda$ play that role instead. 

In order to look at a more complicated example than matrices, but still keep the discussion and calculations manageable, we will only consider tensors of degree 3 (i.e. $K=3$). Since the difficulty of the direct evaluation of $\mathcal{Z}_R(\Delta)$ rapidly increases due to the high dimension of the integration region, we will only focus on low values of $N$. To illustrate: noting that we also have to integrate over the normalised tensorspace, the integration region for generic tensors with $N=3$ for $R=2$ is already  40-dimensional. Considering the derivation in section~\ref{sec:derivation} and the evidence for the existence of $G_R$ presented in section~\ref{sec:existence}, we will only show results for low values of $N$, as sufficient evidence for~\eqref{eq:result_Z_symmetric} is already at hand.

In the symmetric case the $N=2$ case is only well-defined for $R=1$, since $R_c=1$ as can be found in table~\ref{table:Rc_results}. This means that only evaluating $N=2$ would yield only limited insight, hence we also evaluated cases for $N=3$. We evaluated all cases up to $R_c=10$, and found that results always agree with~\eqref{eq:result_Z_symmetric} up to numerical errors. Two examples may be found in figure~\ref{fig:Z_theta_symmetric}. For the generic case the situation is slightly different. For $N=2$ the critical value $R_c=8$, so we can actually expect interesting behaviour in this case already. Hence we solely focus on the $N=2$ case and evaluate the integral up to $R_c=8$. Two examples of this may be found in figure~\ref{fig:Z_theta_symmetric}.

We may conclude that for both the symmetric and generic cases, the numerical results agree perfectly well with the derived equation~\eqref{eq:result_Z_symmetric}, and moreover match the values of $G_R$ determined independently in the numerical manner
explained in section~\ref{sec:existence}. 

\begin{figure}
    \centering
    \begin{minipage}{0.475\textwidth}
        \centering
        \includegraphics[width=0.95\textwidth]{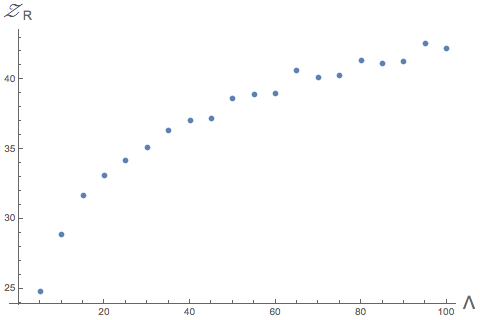}
    \end{minipage}
    \begin{minipage}{0.475\textwidth}
        \centering
        \includegraphics[width=0.95\textwidth]{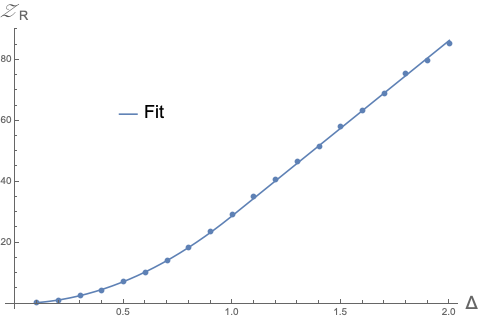}
    \end{minipage}
    \caption{Numerical evaluation of $\mathcal{Z}_{R=2}(\Delta)$ for $N=2$. On the left, we set $\Delta=1$ and vary $\Lambda$ on the horizontal axis. It can be seen that the value indeed diverges linearly, as is expected from the discussion in section~\ref{sec:existence}, since this corresponds to a divergence of $G_R(\epsilon)\propto \epsilon^{-1/2}$ because of $\epsilon \sim \Lambda^{-2}$. On the right, we set $\Lambda=10$ and vary $\Delta$ on the horizontal axis, to show that the functional form (except for the divergent part) is still given by the formula~\eqref{eq:result_Z_symmetric}.}
    \label{fig:Z_theta_generic}
\end{figure}

We finalise this section with a remark on the case of $R>R_c$. In this case $G_R$ diverges and the correctness of formula~\eqref{eq:result_Z_symmetric} is not guaranteed anymore. This leads to a question: Does $\mathcal{Z}_R(\Delta)$ also diverge for $R>R_c$, or is the divergence of $G_R$ only problematic for the derivation of its closed form? We investigated the simplest case for this: symmetric tensors with dimension $N=2$ and rank $R=2$. We found that the $\mathcal{Z}_R(\Delta)$ still diverges by setting $\Delta=1$ and investigating the dependence on $\Lambda$, which can be seen in figure~\ref{fig:Z_theta_generic}. One peculiar fact we discovered is that the functional form of $\mathcal{Z}_R^{\Lambda}(\Delta)$ for fixed and finite $\Lambda$ still follows the functional dependence on $\Delta$ of~\eqref{eq:result_Z_symmetric}, also shown in figure~\ref{fig:Z_theta_generic}. 

This last fact suggests the possibility that the quantity defined in~\eqref{eq:def_ZR_CR} might actually be finite even for $R>R_c$, since the diverging parts will cancel out when taking the $\epsilon\rightarrow0^+$ limit (or $\Lambda\rightarrow\infty$ as in this section). To support this a bit further, let us consider the differential equation solved by the hypergeometric function~\eqref{eq:app:hypergeom_eq}, which is a
homogeneous ordinary differential equation. If we rewrite our result from~\eqref{eq:result_Z_symmetric}\footnote{Here we took the case where $z<1$, the exact same argument holds for the $z>1$ case.}
\begin{equation*}
    {}_2F_1(a,b,c;z) := u(z) \propto z^{-A} \mathcal{Z}_R(z),
\end{equation*}
and plug this into the hypergeometric differential equation, we notice that the resulting equation, which is the equation that $\mathcal{Z}_R(z)$ solves, necessarily still is a homogeneous ordinary differential equation. If we assume that the actual physically relevant properties are described by this differential equation, an overall factor should not matter. Hence, if we extract this overall factor (which might become infinite in the limit $\epsilon\rightarrow0^+$) we should be left with the physically relevant behaviour.

\section{Conclusions and discussions}\label{sec:conclusion}
Motivated by recent progress in the study of the Canonical Tensor Model, we turned our attention in this work to the space of tensor rank decompositions. Because of the analogy between the terms of a tensor rank decomposition and points in a discrete space discussed in~\cite{Kawano:2018pip} we call this the configuration space of tensor rank decompositions. This space has the topology of a product of $R$ times the real line and $R$ times an $N-1$-dimensional unit hemisphere. We equip this space with a measure generated by an infinitesimal volume element, depending on the parameter $w$. In the definition we are rather general, taking into account both symmetric and non-symmetric tensors.

The central result of this work is the derivation of a closed formula for the average volume around a tensor of unit norm, $\mathcal{Z}_R(\Delta)$, in terms of a hypergeometric function in~\eqref{eq:result_Z_symmetric}. This formula depends on the degrees of freedom of the tensor space, the parameter $w$ of the measure, and the rank of the tensor rank decompositions we are considering. The existence of such a closed form formula is far from obvious, and the derivation crucially depends on the existence of a quantity $G_R$. We have investigated the existence of this quantity numerically for the case where $w=1$. In this case the maximum value of $R$ 
for the existence appears to agree with the degrees of freedom of the tensor space $R_c = N_Q$, with the exception of the case for symmetric tensors where $N=2$. Together with earlier results in~\cite{Lionni:2019rty, Sasakura:2019hql, Obster:2020vfo} we conjecture a more general formula~\eqref{eq:conjecture:RC}. Finally we conducted some direct numerical checks for $\mathcal{Z}_R(\Delta)$
and found general agreement with the derived formula.

From a general point of view, we have several interesting future research directions. For one, the conjectured formula~\eqref{eq:conjecture:RC} 
for the maximum $R_c$ is based on the analysis of two values of $w$. It might be worth extending this analysis to more values, which might lead to a more proper analytical explanation for this formula that is currently missing. Secondly, we introduced a quantity $\mathcal{C}_R(\Delta)$, describing the amount of decompositions of size less than $\Delta$. Dividing $\mathcal{Z}_R(\Delta)$ by $\mathcal{C}_R(\Delta)$, we expect that this leads to a meaningful quantity that is finite, even for $R>R_c$. Understanding this quantity and its convergence (or divergence) better would be worth investigating. Finally, a peculiar connection between $w$ and the expected rank was found for some examples, where tuning $w$ as in~\eqref{eq:approximate_Q} lead to $\mathcal{Z}_R(\Delta=1)$ to be minimised for the expected rank of the tensor space. Whether this is just coincidence, or has some deeper meaning, would be interesting to take a closer look at.

Let us briefly discuss what the results mean for the Canonical Tensor Model. The present work provides a first insight into the question how many tensor rank decompositions are close to a given tensor $Q_{abc}$. This might lead to a better understanding into how many ``discrete universes'' of a given size (i.e. amount of points $R$) are close to a tensor. Some work in this area still remains to be done, as we can only give an estimate since we take the average over tensors of size one. 

To conclude, we would like to point out that the formula~\eqref{eq:result_Z_symmetric} could prove to be important in the understanding of the wave function of the Canonical Tensor Model studied in~\cite{Obster:2017pdq, Obster:2017dhx, Lionni:2019rty, Sasakura:2019hql, Obster:2020vfo, Sasakura:2021lub}. In~\cite{Sasakura:2021lub}, the phase of the wave function was analysed in the $Q$-representation, however the amplitude of the wave function is not known. From~\cite{Obster:2017pdq, Obster:2017dhx} we expect that there is a peak structure, where the peaks are located at $Q_{abc}$ that are symmetric under Lie group symmetries. In the present paper we have determined an exact formula for the mean amplitude, which we can use to compare to the local wave function values.

\vspace{.3cm}
\section*{Acknowledgements}
The work of N.S. is supported in part by JSPS KAKENHI Grant No.19K03825. 

\appendix
\section{Tensor Rank Decompositions}\label{app:sec:TRD}
The tensor rank decomposition, also called the canonical polyadic decomposition, may be thought of as a generalisation of the singular value decomposition (SVD) for matrices, which are tensors of degree two, to tensors of general degree. For a more extensive introduction to tensors and the tensor rank decomposition, we would like to refer to~\cite{hackbusch2019tensor, landsberg2011tensors}. 

The SVD decomposes a given real $N\times N$ matrix $M$ into $M = A^T \Lambda B$, where $A$ and $B$ are orthogonal matrices and $\Lambda$ is a diagonal matrix, the diagonal components of which are called the singular values.\footnote{To keep the discussion simple, only real $N\times N$ matrices are considered here, but this may be generalised in a straightforward manner.} The amount of non-zero singular values of a given matrix is called the rank of the matrix, denoted by $R$. To extend the SVD to tensors of general degree, let us rewrite this in a more suggestive form which is called the dyadic notation of the matrix
\begin{equation*}
    M_{ab} = \sum_{i=1}^{R} \sum_{j=1}^{R} (A_{i})_a\ \Lambda_{ii} \delta_{ij}\ (B_{j})_b := \sum_{i=1}^R \lambda_i v_a^i w_b^i, 
\end{equation*}
where $v^i, w^i \in  \mathbb{R}^N$ and $\lambda_i \equiv \Lambda_{ii}\in \mathbb{R}$ are the nonzero singular values. The generalisation to general tensors of degree K is now straightforward:
\begin{equation}
    Q_{a_1\ldots a_K} = \sum_{i=1}^R \lambda_i {v_{a_1}^{(1)}}^i \ldots {v_{a_K}^{(K)}}^i,
\end{equation}
where the rank $R$ is now defined as the lowest number for which such a decomposition exists, and ${v^{(k)}}^i\in\mathbb{R}^N$. For symmetric tensors (similar to symmetric matrices) we can find a decomposition in terms of symmetric rank-1 tensors, meaning that every term in the decomposition is generated by a single vector
\begin{equation*}
    Q_{a_1\ldots a_K} = \sum_{i=1}^R \lambda_i {v_{a_1}}^i \ldots {v_{a_K}}^i.
\end{equation*}
The minimum $R$ for which this is possible is called the symmetric rank.

The space of tensor rank decompositions with $R$ components, $\mathcal{F}_R$, is a subset of the full tensor space
\begin{equation*}
    \mathcal{F}_R \subset \mathcal{T} = V\otimes \ldots \otimes V.
\end{equation*}
This space increases as $R$ becomes bigger, and in its limit it spans the whole tensor space. A typical rank $R_t$ of the tensor space $\mathcal{T}$ is a rank for which $\mathcal{F}_R$ has positive measure in the full tensor space. This typical rank is not necessarily unique, but if this is the case it is called the generic rank. 

The expected generic rank, $R_E$, is a conjectured formula for the generic rank that a tensor space is expected to have, which has been proven to provide a lower estimate of the generic rank. The formula for the non-symmetric case is given by:
\begin{equation}
    R_E = \left \lceil \frac{N^K}{N*K-K+1} \right \rceil.\label{eq:app:expected_rank}
\end{equation}

Note that while the tensor rank decomposition generalises the singular value decomposition, there are many differences between the two~\cite{pierreTRD}. For example, often the tensor rank decomposition is unique~\cite{koldaTRD}, but actually computing the tensor rank decomposition is very hard~\cite{Hillar_NPhard}. 

Note that the vectors ${v^{(k)}}^i$ may be re-scaled as 
\begin{align*}
    {\phi^{(k)}}^i &:= \pm \frac{{v^{(k)}}^i}{\left\|{v^{(k)}}^i\right\|},\\
    \lambda_i &\rightarrow \lambda_i \prod_{k=1}^K\left(\pm \left\|{v^{(k)}}^i\right\|\right),
\end{align*}
where the sign is taken such that ${\phi^{(k)}}^i$ lies on the upper hemisphere $S^{N-1}_+ \subset \mathbb{R}^N$. This is the form we will use in order to remove redundancies in the definition.

\section{Lemmas}\label{app:sec:results}
This appendix section contains two lemmas used in the propositions of section~\ref{sec:derivation}. 

\begin{lemma}\label{lemma:Y_alpha_gamma_finite}
    Given that $G_R$ in~\eqref{eq:def_G} is finite, for $\alpha, \gamma >0 $ the following limit of the integral
    \begin{equation}
        Y(\alpha, \gamma) := \lim_{\epsilon\rightarrow0^+}\int_{\mathbb{R}^{N_Q}} {\rm d}Q\, \int_{\mathcal{F}_R} {\rm d}\Phi \ \e^{- \alpha Q^2 - \gamma (Q - \Phi)^2 -\epsilon\sum_{i=1}^R \lambda_i^2},\label{eq:ZRalphagamma}
    \end{equation}
    is finite.
\end{lemma}
\begin{proof}
    Using the same inequality with $0<A<1$,
    \begin{equation*}
       \|Q-\Phi\|^2 \geq \left(\|Q\| - \|\Phi\|\right)^2 \geq A\|\Phi\|^2 - \frac{A}{1-A}\|Q\|^2,
    \end{equation*}
    as in step one of the proof of proposition~\ref{prop:barZ_result}, we obtain
    \begin{align*}
        Y(\alpha, \gamma) &\leq \lim_{\epsilon\rightarrow0^+}\int_{\mathbb{R}^{N_Q}} {\rm d}Q\, \int_{\mathcal{F}_R} {\rm d}\Phi\ \e^{- \alpha Q^2 + \frac{\gamma A}{1-A}Q^2 - \gamma A \Phi^2  -\epsilon \sum_{i=1}^R \lambda_i^2},\\
        &=\int_{\mathbb{R}^{N_Q}} {\rm d}Q\, \e^{-\left(\alpha-\frac{\gamma A}{1-A}\right) Q^2} \lim_{\epsilon\rightarrow0^+} \int_{\mathcal{F}_R} {\rm d}\Phi\, \e^{-\gamma A \Phi^2 - \epsilon\sum_{i=1}^R \lambda_i^2}.
    \end{align*}
    In the second line it can be seen that the $Q$ and $\Phi$ integration decouple, where the $Q$ integration is simply a finite Gaussian integral if one takes $A$ such that $\alpha > \frac{\gamma A}{1-A}$. The $\Phi$ integration is nothing more than a finite constant multiplied by $G_R$.
    
    Hence, we conclude that this integration is finite if $\lim_{\epsilon\rightarrow 0^+}G_R(\epsilon)$ exists.
\end{proof}
\begin{lemma}\label{lemma:Y_integral_limit}
   The limits in equation~\eqref{eq:Y_alpha_gamma_limits} may be safely interchanged, i.e.
   \begin{equation}\label{eq:app_Y_intergral_limit}
       \lim_{\epsilon\rightarrow0^+} \int_{0}^\infty {\rm d}x\, \bar{Z}_R(x;\epsilon\, x) \, x^{\frac{N_Q+w\,R}{2}-1} \, \e^{-t\,x} = \int_{0}^\infty {\rm d}x \lim_{\epsilon\rightarrow0^+} \bar{Z}_R(x; \epsilon\, x) \, x^{\frac{N_Q+w\, R}{2}-1}  \, \e^{-tx},
   \end{equation}
   under the assumption that $\lim_{\epsilon\rightarrow0^+}G_R(\epsilon)$ converges and is finite.
\end{lemma}
\begin{proof}
    In order to prove~\eqref{eq:app_Y_intergral_limit}, let us take an $X>0$ and split the integral into two parts
    \begin{equation*}
        \lim_{\epsilon\rightarrow0^+} \int_{0}^X {\rm d}x\, \bar{Z}_R(x; \epsilon\, x) \, x^{\frac{N_Q+w\,R}{2}-1} \, \e^{-t\,x} + \lim_{\epsilon\rightarrow0^+}\int_X^\infty {\rm d}x\, \bar{Z}_R(x; \epsilon\, x) \, x^{\frac{N_Q+w\,R}{2}-1} \, \e^{-t\,x},
    \end{equation*}
    and consider both parts separately.
    
    For the first term, we know that the integral and limit can be interchanged if the integrand is uniformly convergent, i.e.  
    \begin{equation*}
        \lim_{\epsilon\rightarrow0^+} \sup_{x\in\left[0,X\right)} \left|x^{\frac{N_Q+w\,R}{2}-1} \, \e^{-t\,x}\left(\bar{Z}_R(x;\epsilon \, x) - \bar{Z}_R(x) \right)\right|=0.
    \end{equation*}
    Now, note that the function $\bar{Z}_R(x; \epsilon \, x)$ is bounded by a contribution proportional to $x^{-\frac{w\,R}{2}}$ as shown in~\eqref{eq:Z_R_bound}, but the expression above has a factor of $x^{\frac{N_Q+w\,R}{2}-1}$ thus the point $x=0$ does not pose a problem and the value above is finite for all $x\in[0,X)$. But then, since from the first step of proposition~\ref{prop:barZ_result} we know $\bar{Z}_R(x; \epsilon\, x) \rightarrow \bar{Z}_R(x)$, 
    \begin{equation*}
        \forall_{x\in[0,X)}\ |x^{\frac{N_Q+w\,R}{2}-1} \e^{-t\,x} | |\bar{Z}_R(x;\epsilon \, x) - \bar{Z}_R(x)| \rightarrow 0,
    \end{equation*}
    and hence we have uniform convergence, meaning that the integral and limiting operations may be interchanged.
    
    For the second term, since $\bar{Z}_R(x;\epsilon)$ is decreasing in $x$ and $\epsilon$, we obtain an upper bound (and using the convergence of $\bar{Z}_R(x; \epsilon \, x)$ which has been proven already)
    \begin{align*}
        \int_X^\infty {\rm d}x\, \bar{Z}_R(x; \epsilon \, x) \, x^{\frac{N_Q+w\,R}{2}-1} \, \e^{-t\,x} &\leq \int_X^\infty {\rm d}x\, \bar{Z}_R(X) \, x^{\frac{N_Q+w\,R}{2}-1} \, \e^{-t\,x},\\
        &= \bar{Z}_R(X) \int_X^\infty {\rm d}x\, x^{\frac{N_Q+w\,R}{2}-1} \e^{-t\,x}.
    \end{align*}
    Now the $\bar{Z}_R(X)$ does not increase for larger $X$, and the final integral converges to zero for large X. This means that the left-hand side vanishes in the limit $X\rightarrow\infty$.
    
    Thus we conclude that the integral and limiting operations may be interchanged.
\end{proof}

\section{Necessary formulae}\label{app:sec:formulae}
In this work we use some nontrivial formulae that are listed in this subsection. Most of them are used in section~\ref{sec:derivation} for the proof of proposition~\ref{prop:barZ_result} and~\ref{prop:Z}. This section is divided in formulas related to the hypergeometric functions, section~\ref{app:sec:hypergeom}, and formulas directly related to the inverse Laplace transforms, section~\ref{app:sec:inverse_laplace}. 

\subsection{Properties of hypergeometric functions}\label{app:sec:hypergeom}
The hypergeometric function and its generalisations play a central role in many fields of mathematics, physics and other sciences. The reason for this is that many of the special functions used throughout these areas can be expressed in terms of the hypergeometric function. An overview of the hypergeometric function and its application may be found in~\cite{seaborn1991hypergeometric}, and a resource for the confluent hypergeometric function (including the Whittaker's function mentioned below) may be found in~\cite{slater1960confluent}. In this work the final result is expressed in terms of the hypergeometric function, whereas in the derivation we use the confluent hypergeometric function. This appendix section summarises some important notions, definitions and formulae.

%
The generalised hypergeometric function, in some sense a generalisation of the geometric series, is defined as the analytic continuation of the series
\begin{equation}\label{eq:app:gen_hypergeometric_series}
   {}_pF_q(a_1,\ldots,a_p,b_1,\ldots,b_q;z) = \sum_{n=0}^{\infty} \frac{(a_1)_n \ldots (a_p)_n}{n! (b_1)_n\ldots(b_q)_n} z^n,
\end{equation}
where we used the Pochhammer symbols
\begin{equation*}
    (a)_n = \frac{\Gamma[a+n]}{\Gamma[a]}.
\end{equation*}
The hypergeometric function is the case where $p=2$ and $q=1$, i.e., inside the range of convergence 
\begin{equation}
    {}_2F_1(a,b,c;z) = \sum_{n=0}^\infty \frac{(a)_n (b)_n}{n! (c)_n} z^n.
\end{equation}
The hypergeometric function may also be defined as the solution to the hypergeometric differential equation
\begin{equation}\label{eq:app:hypergeom_eq}
    z(1-z) \frac{{\rm d}^2 u(z)}{{\rm d}z^2} + [c-(a+b+1)z]\frac{{\rm d}u(z)}{{\rm d}z} - ab\,u(z)=0.
\end{equation}
For ${\rm Re}(c) > {\rm Re}(b) > 0$ and $z$ not being a real number on $z\geq 1$, the hypergeometric function has an integral representation,\footnote{Actually this is the proper analytic continuation of the series above.}
\begin{equation}\label{eq:appendix:hypergeometric_integral}
    {}_2F_1(a,b,c;z) = \frac{1}{\beta(b,c-b)} \int_{0}^1 {\rm d}t\, t^{b-1}(1-t)^{c-b-1}(1-zt)^{-a},
\end{equation}
where $\beta(a,b)$ is the beta-function defined by
\begin{equation}\label{eq:app:beta_f}
    \beta(a,b) := \frac{\Gamma[a]\Gamma[b]}{\Gamma[a+b]}.
\end{equation}

The confluent hypergeometric function is defined by the limit
\begin{equation}
    M(a,c; z) := \lim_{b\rightarrow\infty} {}_2F_1(a,b,c;z/b) = {}_1F_1(a,c;z),
\end{equation}
which exactly corresponds to the series representation defined in~\eqref{eq:app:gen_hypergeometric_series} for $p=q=1$. The differential equation associated to this function may be found in a similar way, and is called the Kummer's equation\footnote{There is another function besides ${}_1F_1(a,c;z)$ that satisfies the differential equation in~\eqref{eq:app:confl_eq}. This is called the confluent hypergeometric function of the second kind.}
\begin{equation}\label{eq:app:confl_eq}
    z \frac{{\rm d}^2 w(z)}{{\rm d}z^2} + [c-z]\frac{{\rm d}w(z)}{{\rm d}z} - a\,w(z)=0.
\end{equation}
The confluent hypergeometric function also has an integral representation given by
\begin{equation}\label{eq:app:confluent_integral}
    {}_1F_1(a,c;z) = \frac{1}{\beta(a,c-a)}\int_{0}^1 {\rm d}t\, \e^{t\,z}\, t^{a-1} (1-t)^{c-a-1},
\end{equation}
for ${\rm Re}(c) > {\rm Re}(a) > 0$. One property of the confluent hypergeometric function we will need is Kummer's transformation:
\begin{equation}\label{eq:app:kummer_theorem}
    \e^{-z}{}_1F_1(a,c;z) = {}_1F_1(c-a, c; -z).
\end{equation}
The Whittaker functions are a variant of both of the confluent hypergeometric functions. The first Whittaker function is the only one we will use and it is defined by~\cite{slater1960confluent}
\begin{equation}\label{eq:app:confl_whittaker}
    M_{\nu,\mu}(z) := \e^{-\frac{z}{2}} z^{\mu+\frac{1}{2}} {}_1F_1(\mu-\nu+\frac{1}{2},1+2\mu; z).
\end{equation}

\subsection{The (inverse) Laplace transform}\label{app:sec:inverse_laplace}
The Laplace transform and its inverse are heavily used tools in mathematics, physics, engineering and other sciences. A good introduction and overview of this area of mathematics is~\cite{doetsch1974introduction}. In~\cite{oberhettinger1973tables}, many explicit Laplace transforms may be found.\footnote{A note of caution here; since the formula for~\eqref{eq:appendix:InverseLaplace_3} for instance is incorrect.}

The Laplace transform (or Laplace integral) of a function $f(t)$ is given by
\begin{equation}\label{eq:app:laplace_transform}
    F(s) \equiv L(f)(s) := \int_{0}^\infty \e^{-s\,t} \, f(t)\, {\rm d}t.
\end{equation}
The Laplace transform is a very useful tool in many aspects. For our purposes on one hand it is possible to convert a complicated integral to a closed formula in the Laplace-space and secondly we find a formula that exactly corresponds to a Laplace transform which lets us extract a function by taking the inverse Laplace transform. Generally it is often used for solving differential equations. The main reason for this is that under the Laplace transformation, taking a derivative corresponds to multiplication by the variable $s$ in the Laplace-space.

Of course, neither taking the Laplace transform nor taking the inverse Laplace transform is always an easy task. In our case, taking the Laplace transform is not that difficult, but the inverse Laplace transform is more involved.

The Laplace transform of a function $f(t)$ exists if the function satisfies two properties: (1) It is of exponential order, (2) it is integrable over any finite domain in $[0,\infty)$. Note that from~\eqref{eq:app:laplace_transform} it can easily be seen that the inverse Laplace transform cannot be unique, since every null-function (a function of measure zero) may be added to a function and result in the same Laplace transform. Hence, the inverse Laplace transformation can only be expected to map towards an equivalence class generated by the null-functions. 
In the present work, however, this ambiguity does not affect our final result: 
the function~\eqref{eq:def_Z_epsilon} is clearly a monotonically increasing function in $\Delta$, and the end-result~\eqref{eq:result_Z_symmetric} is continuous, hence there is no possibility for a null-function to be added.

For two functions $f(t)$ and $g(t)$, we can define the convolution as
\begin{equation}
    (f * g)(t) = \int_0^t f(\tau) g(t-\tau)\, {\rm d}\tau.
\end{equation}
It can straightforwardly be verified that convolution is both commutative and associative. 
If we assume the convergence of the Laplace integral of $f(t)$ and $g(t)$, then the convolution theorem holds
\begin{equation}\label{eq:appendix:ProductFormula}
    L(f * g) = L(f) L(g),
\end{equation}
in other words, the convolution of two functions in the usual domain corresponds to a product in the Laplace domain. 

The Laplace transform used in section~\ref{sec:derivation} is just a straightforward computation of~\eqref{eq:app:laplace_transform}, but we also use two inverse Laplace transforms. Hence, below are three inverse Laplace transformations we use. We will give short proofs for the formulae.

The first inverse Laplace transform we need is a relatively easy one, namely the inverse Laplace transform of $x^{-A-1}$:
\begin{equation}\label{eq:appendix:InverseLaplace_2}
    L^{-1}[x^{-A-1}] = \frac{t^{A}}{\Gamma[A+1]}.
\end{equation}
This can be found by using~\eqref{eq:app:laplace_transform} on the right-hand side. This formula is valid for $A>-1$. 

In this work we need the inverse Laplace transform of $(1+x)^{-A}x^{-B}$. This is given by
\begin{equation}\label{eq:appendix:InverseLaplace_1}
    L^{-1}[(1+x)^{-A} x^{-B}] = t^{-1 + A + B} \frac{{}_1F_1\left(A, A+B, -t\right)}{\Gamma(A+B)}.
\end{equation}
Showing this is a little less trivial. For this, let us take the Laplace transform of the right hand side, using the integral representation of~\eqref{eq:app:confluent_integral},
\begin{align*}
    L\left[t^{-1 + A + B} {}_1F_1\left(A, A+B, -t\right)\right]&=\frac{1}{\beta(A,B)}\int_0^\infty {\rm d}t\, \e^{-t\,x}t^{-1 + A + B} \int_0^1 {\rm d}\tau\, \e^{-t\,\tau}\tau^{A-1}(1-\tau)^{B-1},\\
    &= \frac{1}{\beta(A,B)}\int_0^1 {\rm d}\tau\,\tau^{A-1}(1-\tau)^{B-1} \int_{0}^\infty {\rm d}t\, \e^{-t\,(x+\tau)}t^{-1 + A + B},\\
    &= \frac{\Gamma[A+B]}{\beta(A,B)}\int_0^1 {\rm d}\tau \, \tau^{A-1}(1-\tau)^{B-1}(x+\tau)^{-A-B},\\
    &= \Gamma[A+B] (1+x)^{-A} x^{-B},
\end{align*}
where in the second step we used~\eqref{eq:appendix:InverseLaplace_2}.

The last explicit equation we will need is related to the Whittaker function~\eqref{eq:app:confl_whittaker},
\begin{equation}\label{eq:appendix:InverseLaplace_3}
    L^{-1}\left[ \beta\left({\scriptstyle \mu-\nu+\frac{1}{2}, \mu+\nu+\frac{1}{2}}\right) \, x^{-\frac{1}{2} - \mu} \e^{-\frac{x}{2}} M_{\nu,\mu}(x)  \right] = \begin{cases}
    0, & t<0,\\
    t^{\mu+\nu-\frac{1}{2}} (1-t)^{\mu-\nu-\frac{1}{2}}, & 0 \leq t \leq 1,\\
    0, & t > 1.
  \end{cases}
\end{equation}
0ne can find this inverse Laplace transform by using the definition of the Laplace transfrom~\eqref{eq:app:laplace_transform}, the integral representation of the confluent hypergeometric function~\eqref{eq:app:confluent_integral}, the definition of the Whittaker function~\eqref{eq:app:confl_whittaker}, and Kummer's transformation~\eqref{eq:app:kummer_theorem}:
\begin{align*}
    L\left[t^{\mu+\nu-\frac{1}{2}}(1-t)^{\mu-\nu-\frac{1}{2}} \Theta(t < 1)\right] &= \int_{0}^1 {\rm d}t \, \e^{-x\,t}\, t^{\mu+\nu-\frac{1}{2}}(1-t)^{\mu-\nu-\frac{1}{2}},\\
    &= \beta\left({ \scriptstyle\mu-\nu+\frac{1}{2}, \mu+\nu+\frac{1}{2}}\right) \, {}_1F_1(\mu+\nu+\frac{1}{2}, 2\mu+1; -x),\\
    &= \beta\left({\scriptstyle \mu-\nu+\frac{1}{2}, \mu+\nu+\frac{1}{2}}\right) \, \e^{-x}{}_1F_1(\mu-\nu+\frac{1}{2}, 2\mu+1; x),\\
    &= \beta\left({\scriptstyle \mu-\nu+\frac{1}{2}, \mu+\nu+\frac{1}{2}}\right) \, x^{-\frac{1}{2}-\mu} \e^{-\frac{x}{2}} M_{\nu,\mu}(x).
\end{align*}

\section{The expression of \texorpdfstring{${\mathcal C}_R(\Delta)$}{CR(Delta)}}\label{sec:app:delta_dep}
In \eqref{eq:defofcr} we introduce the following quantity:
\begin{equation*}
    \mathcal{C}_R(\Delta) := \int_{\mathcal{F}_R} {\rm d} \Phi_w\, \Theta(\Delta - \|\Phi\|^2).
\end{equation*}
A proper definition of this quantity would assume a regularisation function like in~\eqref{eq:def_Z_epsilon}. In this appendix section we keep the discussion short and heuristic. A proper derivation including this regularisation function would go exactly along the lines of the derivation of $\mathcal{Z}_R(\Delta)$ in section~\ref{sec:derivation}. In a similar way as the derivation of $\mathcal{Z}_R(\Delta)$, assuming the existence of $G_R$, we can now take the Laplace tranform 
\begin{align*}
    \bar{\mathcal{C}}_R(\gamma) &= \int_0^\infty {\rm d}\Delta \, \int_{\mathcal{F}_R} {\rm d}\Phi_w \, \e^{-\gamma \Delta} \Theta(\Delta - \|\Phi\|^2),\\
    &= \int_{\mathcal{F}_R} {\rm d}\Phi_w \, \int_{\|\Phi\|^2}^{\infty} {\rm d}\Delta \, \e^{-\gamma\Delta},\\
    &= \gamma^{-1}\int_{\mathcal{F}_R} {\rm d}\Phi_w\, \e^{-\gamma \Phi^2},\\
    &= \gamma^{-\frac{w\, R}{2}-1}\, G_R.
\end{align*}
Now that we related the Laplace transform to $G_R$, we can take the inverse Laplace transform, using~\eqref{eq:appendix:InverseLaplace_2}:
\begin{equation}
    \mathcal{C}_R(\Delta) = \frac{G_R}{\Gamma\left[\frac{w\, R}{2} + 1\right]} \Delta^{\frac{w\, R}{2}}.
\end{equation}

\printbibliography

@book{doetsch1974introduction,
  title={Introduction to the Theory and Application of the Laplace Transformation},
  author={Doetsch, G. and Nader, W.},
  isbn={9780387064079},
  lccn={74185469},
  year={1974},
  publisher={Springer-Verlag}
}

@book{oberhettinger1973tables,
  title={Tables of Laplace Transforms},
  author={Oberhettinger, F. and Badii, L.},
  isbn={9780387063508},
  lccn={73081328},
  year={1973},
  publisher={Springer-Verlag}
}

@book{seaborn1991hypergeometric,
  title={Hypergeometric Functions and Their Applications},
  author={Seaborn, J.B.},
  isbn={9783540975588},
  lccn={lc91018910},
  series={Texts in Applied Mathematics},
  year={1991},
  publisher={Springer}
}

@book{slater1960confluent,
  title={Confluent Hypergeometric Functions},
  author={Slater, L.J.},
  isbn={9780608308982},
  lccn={60004198},
  year={1960},
  publisher={University Press}
}

@article{Kawano:2018pip,
    author = "Kawano, Taigen and Obster, Dennis and Sasakura, Naoki",
    title = "{Canonical tensor model through data analysis: Dimensions, topologies, and geometries}",
    eprint = "1805.04800",
    archivePrefix = "arXiv",
    primaryClass = "hep-th",
    reportNumber = "YITP-18-38",
    doi = "10.1103/PhysRevD.97.124061",
    journal = "Phys. Rev. D",
    volume = "97",
    number = "12",
    pages = "124061",
    year = "2018"
}

@article{Sasakura:2011sq,
      author         = "Sasakura, Naoki",
      title          = "{Canonical tensor models with local time}",
      journal        = "Int. J. Mod. Phys.",
      volume         = "A27",
      year           = "2012",
      pages          = "1250020",
      doi            = "10.1142/S0217751X12500200",
      eprint         = "1111.2790",
      archivePrefix  = "arXiv",
      primaryClass   = "hep-th",
      reportNumber   = "YITP-11-93",
      SLACcitation   = "%%CITATION = ARXIV:1111.2790;%%"
}

@article{Sasakura:2012fb,
      author         = "Sasakura, Naoki",
      title          = "{Uniqueness of canonical tensor model with local time}",
      journal        = "Int. J. Mod. Phys.",
      volume         = "A27",
      year           = "2012",
      pages          = "1250096",
      doi            = "10.1142/S0217751X12500960",
      eprint         = "1203.0421",
      archivePrefix  = "arXiv",
      primaryClass   = "hep-th",
      reportNumber   = "YITP-12-12",
      SLACcitation   = "%%CITATION = ARXIV:1203.0421;%%"
}

@article{Sasakura:2014gia,
      author         = "Sasakura, Naoki and Sato, Yuki",
      title          = "{Interpreting canonical tensor model in minisuperspace}",
      journal        = "Phys. Lett.",
      volume         = "B732",
      year           = "2014",
      pages          = "32-35",
      doi            = "10.1016/j.physletb.2014.03.006",
      eprint         = "1401.2062",
      archivePrefix  = "arXiv",
      primaryClass   = "hep-th",
      SLACcitation   = "%%CITATION = ARXIV:1401.2062;%%"
}

@article{Sasakura:2015pxa,
      author         = "Sasakura, Naoki and Sato, Yuki",
      title          = "{Constraint algebra of general relativity from a formal
                        continuum limit of canonical tensor model}",
      journal        = "JHEP",
      volume         = "10",
      year           = "2015",
      pages          = "109",
      doi            = "10.1007/JHEP10(2015)109",
      eprint         = "1506.04872",
      archivePrefix  = "arXiv",
      primaryClass   = "hep-th",
      reportNumber   = "YITP-15-52, WITS-MITP-014",
      SLACcitation   = "%%CITATION = ARXIV:1506.04872;%%"
}

@article{Chen:2016ate,
      author         = "Chen, Hua and Sasakura, Naoki and Sato, Yuki",
      title          = "{Equation of motion of canonical tensor model and
                        Hamilton-Jacobi equation of general relativity}",
      journal        = "Phys. Rev.",
      volume         = "D95",
      year           = "2017",
      number         = "6",
      pages          = "066008",
      doi            = "10.1103/PhysRevD.95.066008",
      eprint         = "1609.01946",
      archivePrefix  = "arXiv",
      primaryClass   = "hep-th",
      reportNumber   = "YITP-16-99",
      SLACcitation   = "%%CITATION = ARXIV:1609.01946;%%"
}

@article{Sasakura:2013wza,
      author         = "Sasakura, Naoki",
      title          = "{Quantum canonical tensor model and an exact wave
                        function}",
      journal        = "Int. J. Mod. Phys.",
      volume         = "A28",
      year           = "2013",
      pages          = "1350111",
      doi            = "10.1142/S0217751X1350111X",
      eprint         = "1305.6389",
      archivePrefix  = "arXiv",
      primaryClass   = "hep-th",
      reportNumber   = "YITP-13-36",
      SLACcitation   = "%%CITATION = ARXIV:1305.6389;%%"
}

@article{Narain:2014cya,
      author         = "Narain, Gaurav and Sasakura, Naoki and Sato, Yuki",
      title          = "{Physical states in the canonical tensor model from the
                        perspective of random tensor networks}",
      journal        = "JHEP",
      volume         = "01",
      year           = "2015",
      pages          = "010",
      doi            = "10.1007/JHEP01(2015)010",
      eprint         = "1410.2683",
      archivePrefix  = "arXiv",
      primaryClass   = "hep-th",
      reportNumber   = "WITS-CTP-150, YITP-14-73",
      SLACcitation   = "%%CITATION = ARXIV:1410.2683;%%"
}

@article{Obster:2017pdq,
    author = "Obster, Dennis and Sasakura, Naoki",
    title = "{Symmetric configurations highlighted by collective quantum coherence}",
    eprint = "1704.02113",
    archivePrefix = "arXiv",
    primaryClass = "hep-th",
    reportNumber = "YITP-17-34",
    doi = "10.1140/epjc/s10052-017-5355-y",
    journal = "Eur. Phys. J. C",
    volume = "77",
    number = "11",
    pages = "783",
    year = "2017"
}

@article{Obster:2017dhx,
    author = "Obster, Dennis and Sasakura, Naoki",
    title = "{Emergent symmetries in the canonical tensor model}",
    eprint = "1710.07449",
    archivePrefix = "arXiv",
    primaryClass = "hep-th",
    reportNumber = "YITP-17-107",
    doi = "10.1093/ptep/pty038",
    journal = "PTEP",
    volume = "2018",
    number = "4",
    pages = "043A01",
    year = "2018"
}

@article{Lionni:2019rty,
    author = "Lionni, Luca and Sasakura, Naoki",
    title = "{A random matrix model with non-pairwise contracted indices}",
    eprint = "1903.05944",
    archivePrefix = "arXiv",
    primaryClass = "hep-th",
    reportNumber = "YITP-19-17",
    doi = "10.1093/ptep/ptz057",
    journal = "PTEP",
    volume = "2019",
    number = "7",
    pages = "073A01",
    year = "2019"
}

@article{Sasakura:2019hql,
    author = "Sasakura, Naoki and Takeuchi, Shingo",
    title = "{Numerical and analytical analyses of a matrix model with non-pairwise contracted indices}",
    eprint = "1907.06137",
    archivePrefix = "arXiv",
    primaryClass = "hep-th",
    reportNumber = "YITP-19-59",
    doi = "10.1140/epjc/s10052-019-7591-9",
    journal = "Eur. Phys. J. C",
    volume = "80",
    number = "2",
    pages = "118",
    year = "2020"
}

@article{Obster:2020vfo,
    author = "Obster, Dennis and Sasakura, Naoki",
    title = "{Phases of a matrix model with non-pairwise index contractions}",
    eprint = "2004.03152",
    archivePrefix = "arXiv",
    primaryClass = "hep-th",
    reportNumber = "YITP-20-34",
    doi = "10.1093/ptep/ptaa085",
    journal = "PTEP",
    volume = "2020",
    number = "7",
    pages = "073B06",
    year = "2020"
}

@article{Sasakura:2021lub,
    author = "Sasakura, Naoki",
    title = "{Phase profile of the wave function of canonical tensor model and emergence of large spacetimes}",
    eprint = "2104.11845",
    archivePrefix = "arXiv",
    primaryClass = "hep-th",
    reportNumber = "YITP-21-36",
    month = "4",
    year = "2021"
}

@article{Hitchcock1927,
author = {Hitchcock, Frank L.},
title = {The Expression of a Tensor or a Polyadic as a Sum of Products},
journal = {Journal of Mathematics and Physics},
volume = {6},
number = {1-4},
pages = {164-189},
doi = {10.1002/sapm192761164},
%url = {https://onlinelibrary.wiley.com/doi/abs/10.1002/sapm192761164},
%eprint = {https://onlinelibrary.wiley.com/doi/pdf/10.1002/sapm192761164},
year = {1927}
}

@article{Hillar_NPhard, 
author = {Hillar, Christopher J. and Lim, Lek-Heng}, 
title = {Most Tensor Problems Are NP-Hard}, 
year = {2013}, 
issue_date = {November 2013}, 
publisher = {Association for Computing Machinery}, 
address = {New York, NY, USA}, 
volume = {60}, 
number = {6}, 
issn = {0004-5411}, 
%url = {https://doi.org/10.1145/2512329}, 
doi = {10.1145/2512329}
}

@book{hackbusch2019tensor,
  title={Tensor Spaces and Numerical Tensor Calculus},
  author={Hackbusch, W.},
  isbn={9783030355548},
  series={Springer Series in Computational Mathematics},
  year={2019},
  publisher={Springer International Publishing}
}

@book{landsberg2011tensors,
  title={Tensors: Geometry and Applications: Geometry and Applications},
  author={Landsberg, J.M.},
  isbn={9780821869079},
  lccn={2011021758},
  series={Graduate studies in mathematics},
  year={2011},
  publisher={American Mathematical Society}
}

@ARTICLE{pierreTRD,
  author={Comon, Pierre},
  journal={IEEE Signal Processing Magazine}, 
  title={Tensors : A brief introduction}, 
  year={2014},
  volume={31},
  number={3},
  pages={44-53},
  doi={10.1109/MSP.2014.2298533}
  }

@article{koldaTRD, 
author = {Kolda, Tamara G. and Bader, Brett W.}, 
title = {Tensor Decompositions and Applications}, 
year = {2009}, 
issue_date = {August 2009}, 
publisher = {Society for Industrial and Applied Mathematics}, 
address = {USA}, 
volume = {51}, 
number = {3},
issn = {0036-1445},
doi = {10.1137/07070111X}
}

@MISC{eigenweb,
  author = {Ga\"{e}l Guennebaud and Beno\^{i}t Jacob and others},
  title = {Eigen v3},
  howpublished = {http://eigen.tuxfamily.org},
  year = {2010}
}

\end{document}